\documentclass[11pt,a4paper]{article}
\usepackage[T1]{fontenc}
\usepackage{lmodern,amsmath,amssymb,amsthm,mathtools}
\usepackage[margin=27mm]{geometry}
\usepackage{booktabs,array,enumitem,microtype}
\usepackage{xcolor}
\usepackage[colorlinks=true,linkcolor=blue!55!black,citecolor=blue!55!black,urlcolor=blue!55!black]{hyperref}
\setlist{itemsep=2pt,topsep=4pt}
\newtheorem{theorem}{Theorem}[section]
\newtheorem{proposition}[theorem]{Proposition}
\newtheorem{lemma}[theorem]{Lemma}
\newtheorem{corollary}[theorem]{Corollary}
\theoremstyle{definition}
\newtheorem{definition}[theorem]{Definition}
\newtheorem{example}[theorem]{Example}

\theoremstyle{remark}
\newtheorem{remark}[theorem]{Remark}
\newcommand{\F}{\mathbb F}
\newcommand{\Z}{\mathbb Z}
\newcommand{\Tr}{\operatorname{Tr}}
\newcommand{\ord}{\operatorname{ord}}
\newcommand{\PAut}{\operatorname{PAut}}
\newcommand{\wt}{\operatorname{wt}}
\newcommand{\PG}{\operatorname{PG}}
\newcommand{\lcm}{\operatorname{lcm}}
\DeclareMathOperator{\rank}{rank}
\DeclareMathOperator{\supp}{supp}
\numberwithin{equation}{section}
\allowdisplaybreaks[1]
\title{Length Bounds for Irreducible Cyclic MDS Codes}
\author{Guanghui~Zhang \footnote{E-mail address: {\it zghui@squ.edu.cn.}}}

\date{\small (School of Mathematics and Physics, Suqian University, Suqian 223800, China)}

\begin{document}
\maketitle
\begin{abstract}
We prove the MDS length bound for nontrivial irreducible cyclic codes of every composite length.
The classical trace realization turns the MDS property into independence of elements of a multiplicative subgroup. 
Every nonidentity subgroup must have the full extension degree. This restriction yields length congruences, 
and the Frobenius-orbit description of the coefficient field gives an exact descent obstruction. 
A cyclotomic argument reduces each fixed dimension to finitely many parameters. For dimensions at most thirty, 
established geometric bounds reduce the alphabet to a small finite range; 
the remaining twelve parameter triples are excluded by explicit products of at most three determinant factors and a low-degree multiplier argument. 
An independent argument proves the composite-length theorem in arbitrary dimension.
The standardness conclusion is stated separately as a conditional consequence of the classification statement of Feng, Hollmann, Li, and Xiang; 
the correct extension-field reduction retains a semilinear quotient. The general prime-length problem is not settled here. 
Neither length proof uses the automorphism classification or the recent odd-characteristic range announced by Fan.
\end{abstract}

\section{Introduction}
\subsection{Background and motivation}
A linear $[n,k,d]_q$ code is a $k$-dimensional subspace of $\F_q^n$ with minimum Hamming distance $d$. The Singleton bound is
\begin{equation}\label{eq:singleton-intro}
 d\leq n-k+1.
\end{equation}
A code attaining equality is \emph{maximum distance separable}, or MDS \cite{Singleton}. For a generator matrix $G$, the MDS property has the equivalent formulation
\[
 \det G_I\neq0\qquad
 (I\subseteq\{0,\ldots,n-1\},\ |I|=k),
\]
where $G_I$ consists of the columns indexed by $I$. Thus every $k$ coordinates determine the whole codeword. The extremal problem is to bound $n$ in terms of $q$ and $k$.

To formulate the length problem beyond the elementary endpoint families, throughout this paper \emph{nontrivial} means $2\leq k\leq n-2$. The linear MDS conjecture predicts
\begin{equation}\label{eq:mds-conjecture}
 n\leq q+1,
\end{equation}
with the possible exception $n=q+2$ when $q$ is even and $k\in\{3,q-1\}$. Repetition codes and their duals explain why nontriviality is necessary. The exceptional dimensions are dual to one another and arise from hyperovals in finite projective planes.

To make the geometric and structural inputs below unambiguous, all numbered theorem and section references to Ball and Lavrauw \cite{BallLavrauw} use the numbering of \href{https://arxiv.org/abs/1908.10772v2}{arXiv:1908.10772v2}. The versioned preprints in the bibliography likewise fix the numbering used for Feng--Hollmann--Li--Xiang, Li--Yuan, and Fan.

The orthogonal-array and arc formulations are classical; see Bush and Segre \cite{Bush,Segre} and the modern account \cite[Theorem~17]{BallLavrauw}. A point of $\PG(k-1,q)$ is a one-dimensional subspace of $\F_q^k$; an \emph{arc} is a set of points for which every $k$ representatives are independent. Consequently,
\[
 \text{MDS generator matrix}
 \quad\longleftrightarrow\quad
 \text{arc in }\PG(k-1,q),
\]
where changing representatives amounts to nonzero column scalings. This correspondence supplies the geometric length bounds used below; see \cite{BallLavrauw}. Our additional input is that, for an irreducible cyclic code, the columns admit a multiplicative description in an extension field.

\subsection{Related work}
The conjecture is known in several ranges that provide the starting points for our reductions. Ball proved the MDS conjecture over prime fields \cite[Theorem~1.7]{Ball}. The elementary coding-theoretic facts used below are standard; we retain their proofs and cite the relevant treatments in \cite{HP,LN}. For $q=p^f$ with $f>1$, Ball and De Beule established the range $k\leq2p-2$ \cite[Theorem~5.2]{BallDeBeule}. Geometric results of Hirschfeld--Korchm\'aros, Storme--Thas, and Voloch supply further ranges \cite{HK,ST,Voloch}. Section~\ref{sec:composite} states precisely the versions used here, including the strict endpoints and excluded field sizes.

A different approach restricts the coordinate symmetry. Deng, Feng, and Vasil'ev prove that a linear MDS code with $4\leq k\leq q-3$ and a two-transitive permutation automorphism group has $n\leq q+1$ \cite[Theorem~1.3]{DFV}. They propose cyclic symmetry as a weaker hypothesis to investigate \cite[Section~5]{DFV}. A regular cyclic action is transitive, but need not be two-transitive; the length problem therefore requires an additional argument.

Motivated by this distinction between transitivity and two-transitivity, we study \emph{irreducible cyclic codes}, meaning that the \emph{check polynomial}
\[
 h(X)=\frac{X^n-1}{g(X)}
\]
is irreducible, where $g$ generates the cyclic ideal. Thus $k=\deg h$, not $\deg g$. After establishing full coordinate period, we identify a generator matrix with the elements of a subgroup $H\leq\F_{q^k}^{\times}$ of order $n$. This is a full cyclic orbit, rather than an arbitrary truncated companion orbit.

Within this irreducible class, permutation-group structure provides a separate line of inquiry. Theorem~1.1 of the preprint of Feng, Hollmann, Li, and Xiang states a classification of nonstandard nondegenerate irreducible cyclic codes \cite{FHLX}. Section~\ref{sec:standard} proves what its classification statement would imply for the permutation group. Because the cited proof contains an unresolved semilinear-to-linear step, this consequence is explicitly conditional; the local correction given here is not a verification of the whole classification. In a related geometric direction, Li and Yuan give a companion-orbit obstruction to rational-normal-curve containment \cite[Corollary~4.6]{LiYuan}; this is separate from the length arguments developed here. In dimension two our elementary criterion also agrees with the irreducible recurrence-code case of Shi, Zhang, and Sol\'e \cite[Theorems~1 and~5]{SZS}.

For comparison with recent general MDS bounds, a preprint of Fan \cite[Theorem~1.1]{Fan} states the odd-characteristic range
\[
\begin{gathered}
 2\leq k\leq B(p,q)\quad\hbox{or}\quad q+2-B(p,q)\leq k\leq q,\\
 B(p,q)=\left\lfloor\frac{(p-2)q+6p-10}{2p-3}\right\rfloor.
\end{gathered}
\]
It covers the eleven odd-characteristic triples in Table~\ref{tab:residual}. It also implies the odd-characteristic composite-length assertion after subgroup rigidity: for $p\geq3$ and $q\geq k^2+k+2$,
\[
 \frac{(p-2)q+6p-10}{2p-3}-\frac{q+8}{3}
 =\frac{(p-3)(q+2)}{3(2p-3)}\geq0,
 \qquad \frac{k^2+k+10}{3}\geq k.
\]
Thus $B(p,q)\geq k$. The length proofs below do not use this recent preprint: they retain independent arguments from the earlier geometric bounds and explicit polynomial arithmetic. Accordingly, the overlapping odd-characteristic length conclusions are not presented as conclusions absent from the current literature.

Table~\ref{tab:comparison} places these external results alongside the deductions developed here. In particular, neither the automorphism classification nor the companion-orbit obstruction is an input to either length theorem.
\begin{table}[htbp]
\centering\small
\caption{External inputs and the conclusions obtained here.}\label{tab:comparison}
\begin{tabular}{@{}>{\raggedright\arraybackslash}p{.24\textwidth}>{\raggedright\arraybackslash}p{.34\textwidth}>{\raggedright\arraybackslash}p{.36\textwidth}@{}}
\toprule
Source & Scope or conclusion & Role in this paper\\
\midrule
Ball \cite{Ball}; Ball--De Beule \cite{BallDeBeule}
& Prime alphabets; $k\leq2p-2$ over nonprime alphabets.
& Remove these cases before the arithmetic sieve.\\[3pt]
Hirschfeld--Korchm\'aros, Storme--Thas, Voloch \cite{HK,ST,Voloch}
& Geometric bounds in specified ranges of $q,k$.
& Combine with subgroup rigidity to handle all composite lengths.\\[3pt]
Deng--Feng--Vasil'ev \cite{DFV}
& Two-transitive MDS codes, $4\leq k\leq q-3$.
& Motivation; not applicable under the conditional standardness conclusion.\\[3pt]
Feng--Hollmann--Li--Xiang \cite{FHLX}
& Classification of nonstandard irreducible cyclic codes.
& Conditional structural consequence only; see Remark~\ref{rem:classification-audit}.\\[3pt]
Li--Yuan \cite{LiYuan}
& Companion-orbit obstruction to GRS type.
& Related structural result; not used in either length proof.\\[3pt]
Fan \cite{Fan}
& A linear-in-$q$ range in odd characteristic.
& Covers overlapping odd-characteristic cases; not used in our proofs.\\[3pt]
Arguments developed here
& Subgroup restrictions, exact coefficient-field descent, and polynomial reduction.
& Independent length proofs; unrestricted congruence bounds; explicit separation of structural dependencies.\\
\bottomrule
\end{tabular}
\end{table}

\subsection{Contributions and main results}
The distinction between established facts and deductions is used throughout. The trace and column descriptions, elementary MDS operations, Frobenius-orbit coefficient-field argument, and dimension-two criterion have the references indicated at their statements, with their proofs retained. The factor-count, congruence, normality, and high-dimension restrictions are proved as consequences of these tools and subgroup rigidity. The automorphism conclusion is conditional. No assertion of priority for every elementary reformulation is intended.

The deductions begin with a restriction on the column orbit: every nonidentity subgroup must span the whole extension field. In multiplicative-order notation this becomes
\begin{equation}\label{eq:intro-rigidity}
 \ord_d(q)=k\qquad(d\mid n,\ d>1).
\end{equation}
Indeed, a subgroup of degree $r<k$ supplies $r+1\leq k$ dependent columns. For every prime $\ell\mid n$, this gives
\[
 k\mid\ell-1,
 \qquad n\equiv1\pmod k,
 \qquad n\text{ odd}.
\]
Thus the even exceptional length $q+2$ cannot occur in our class. Intersecting the length congruence with the elementary Singleton-type interval also gives
\begin{align}
 k\mid(q-1)&\ \Longrightarrow\ n\leq q-k,\label{eq:intro-family1}\\
 k\mid q,\quad k\geq3&\ \Longrightarrow\ n\leq q-k+1.\label{eq:intro-family2}
\end{align}
These two implications have no upper restriction on $k$.

Subgroup rigidity constrains the length; field descent constrains the alphabet. Writing $q=p^f$, we show that a hypothetical counterexample $n>q+1$ must have
\[
 \ord_n(p)=fk.
\]
Finally, with $t=n-q$, the counterexample interval and the subgroup restrictions imply
\[
 2\leq t\leq k-1,
 \qquad n\mid\lvert\Phi_k(-t)\rvert.
\]
For fixed $k$ this is a finite problem. For $k\leq30$, the geometric bounds also restrict $q$ to a small finite range. The surviving parameters are excluded by the coefficient matrix of a polynomial multiple of degree at most $k+3$.

The resulting length statements are collected in the following theorem.
\begin{theorem}\label{thm:main}
Every nontrivial irreducible cyclic MDS code of composite length satisfies $n\leq q+1$. The same conclusion holds at every length when its dimension is at most thirty. In either case, if the dimension is at least three, then $n\leq q-1$.
\end{theorem}
The two parts of the theorem use different final reductions. The composite-length proof combines the subgroup and descent restrictions with the cited geometric bounds and a small integer enumeration. It does not use the polynomial calculation for the twelve prime-length cases. Section~\ref{sec:composite} gives this independent chain of implications in detail.

The permutation-group conclusion complements these length bounds. Assuming the classification statement of \cite[Theorem~1.1]{FHLX}, we obtain
\[
 \PAut(C)\cong C_n\rtimes_q C_k,
\]
and this group is not two-transitive for a nontrivial code. This structural conclusion is conditional on the cited classification, whereas both length assertions are independent of it. The prime-length problem above dimension thirty remains unresolved by the present arguments.

\subsection{Organization}
Section~\ref{sec:model} develops the linear-algebraic criteria and the trace realization, including full coordinate period. Section~\ref{sec:rigidity} proves subgroup rigidity and its immediate consequences. Section~\ref{sec:standard} treats the conditional permutation-group result and its consequence for two-transitivity. Sections~\ref{sec:descent} and~\ref{sec:finite} establish exact descent and the finite-dimensional length theorem. Section~\ref{sec:composite} proves the composite-length theorem.

\section{Linear-algebraic preliminaries and the trace model}\label{sec:model}
We first put the order calculations and the column criterion in a form that can be reused without implicit changes of field. We then derive the trace representation from the definition of an irreducible cyclic code. This order of presentation also makes the full-period argument independent of any unproved trace-model assumption.

\subsection{Finite fields and multiplicative orders}
Write $q=p^f$, with $p$ prime and $f\geq1$. All finite fields are regarded as subfields of a fixed algebraic closure of $\F_p$. For positive integers $m$ and $u$ with $\gcd(u,m)=1$, put
\[
 \ord_m(u)=\min\{j\geq1:u^j\equiv1\pmod m\}.
\]
We use $\varphi(m)=|(\Z/m\Z)^\times|$, the valuation $v_\ell(m)$, and
\[
 \Omega(m)=\sum_{\ell\mid m}v_\ell(m).
\]
The cyclotomic polynomials are characterized by
\begin{equation}\label{eq:cyclotomic-product}
 X^m-1=\prod_{d\mid m}\Phi_d(X),
 \qquad \Phi_d(X)\in\Z[X]\text{ monic}.
\end{equation}
Here and below, $\ell\mid m$ under a sum or product over primes means that $\ell$ is prime. For vectors, $\langle\cdot\rangle_K$ denotes linear span over $K$; for an element of a multiplicative group, $\langle x\rangle$ denotes its cyclic subgroup.

The following standard order and Frobenius identities will account for every change of alphabet in the proofs; see \cite[Sections~3.5--3.7]{HP} and \cite[Chapter~2]{LN}. Their elementary proofs are included to fix conventions.
\begin{lemma}\label{lem:orders}
Suppose that the indicated bases and moduli are coprime.
\begin{enumerate}[label=(\roman*)]
\item If $e=\ord_m(u)$ and $a\geq1$, then
\begin{equation}\label{eq:order-power}
 \ord_m(u^a)=\frac{e}{\gcd(e,a)}.
\end{equation}
\item If $d\mid m$, then $\ord_d(u)\mid\ord_m(u)$; moreover,
\begin{equation}\label{eq:order-lcm}
 \ord_{\lcm(m_1,m_2)}(u)
 =\lcm\bigl(\ord_{m_1}(u),\ord_{m_2}(u)\bigr).
\end{equation}
\item If $x$ has multiplicative order $d>1$ in a finite extension of $\F_q$, then
\begin{equation}\label{eq:degree-order}
 [\F_q(x):\F_q]=\ord_d(q).
\end{equation}
\end{enumerate}
\end{lemma}
\begin{proof}
For (i), write $g=\gcd(e,a)$, $e=ge_0$, $a=ga_0$, with $\gcd(e_0,a_0)=1$. For every $j\geq1$,
\[
 (u^a)^j\equiv1\pmod m
 \iff e\mid aj
 \iff e_0\mid a_0j
 \iff e_0\mid j.
\]
The smallest such $j$ is $e_0$.

For (ii), reduction modulo $d$ gives $u^{\ord_m(u)}\equiv1\pmod d$, proving the divisibility. Also
\[
 u^j\equiv1\pmod{\lcm(m_1,m_2)}
 \iff
 \begin{cases}
  u^j\equiv1\pmod{m_1},\\
  u^j\equiv1\pmod{m_2},
 \end{cases}
\]
so the least possible $j$ is the least common multiple of the two orders.

For (iii), membership in a finite subfield is given by
\[
 x\in\F_{q^r}
 \iff x^{q^r}=x
 \iff x^{q^r-1}=1
 \iff d\mid q^r-1.
\]
The least such $r$ is both the degree of $x$ over $\F_q$ and $\ord_d(q)$. Equivalently, this is the size of its Frobenius orbit.
\end{proof}

\subsection{The MDS column criterion, duality, and shortening}
With these field identities in place, we turn to the linear-algebraic properties that will translate the MDS condition into restrictions on the coordinate subgroup.
For $c=(c_0,\ldots,c_{n-1})\in\F_q^n$, define
\[
 \supp(c)=\{i:c_i\neq0\},\qquad
 \wt(c)=|\supp(c)|,\qquad
 d(C)=\min_{0\neq c\in C}\wt(c).
\]
A generator matrix $G\in\F_q^{k\times n}$ has rank $k$ and row space $C$. For a set $I$ of coordinate indices, $G_I$ denotes its column submatrix. If $|I|=k$ and $G_I$ is invertible, then $I$ is an \emph{information set}, because
\[
 c=uG\quad\Longrightarrow\quad u=c_I G_I^{-1}.
\]

To relate information sets to minimum distance, we first recall the Singleton bound. Deleting $d-1$ coordinates is injective on a code of minimum distance $d$: two words with the same remaining coordinates would have a nonzero difference of weight at most $d-1$. Hence
\[
 q^k\leq q^{n-d+1}
 \quad\Longrightarrow\quad d\leq n-k+1.
\]
The equality case can now be expressed entirely by ranks. This is the classical MDS column criterion \cite[Section~2.4]{HP}, equivalently the arc correspondence \cite[Theorem~17]{BallLavrauw}; the rank formula and proof are retained below.
\begin{lemma}\label{lem:general-column}
For a linear $[n,k]_q$ code $C$ with generator matrix $G$, the following are equivalent:
\begin{enumerate}[label=(\roman*)]
\item $C$ is MDS;
\item $\det G_I\neq0$ for every $|I|=k$;
\item every set of at most $k$ columns of $G$ is independent.
\end{enumerate}
More generally,
\begin{equation}\label{eq:distance-rank}
 d(C)=n-\max\{|I|:\rank G_I<k\}.
\end{equation}
\end{lemma}
\begin{proof}
For every coordinate set $I$,
\begin{align*}
 \rank G_I<k
 &\iff \exists\,u\in\F_q^{1\times k}\setminus\{0\}:uG_I=0\\
 &\iff \exists\,c=uG\in C\setminus\{0\}:c_i=0\quad(i\in I).
\end{align*}
Here $uG\neq0$ because $G$ has full row rank. Maximizing the possible number of zero coordinates proves \eqref{eq:distance-rank}. Therefore
\[
 \bigl(\det G_I\neq0\text{ for all }|I|=k\bigr)
 \iff d(C)\geq n-k+1
 \iff d(C)=n-k+1.
\]
The last equivalence uses the Singleton bound. Finally, a dependent set of $s<k$ columns can be enlarged to a dependent set of $k$ columns. This proves (ii)$\Leftrightarrow$(iii).
\end{proof}

Duality and shortening are the two operations needed for elementary length bounds. These standard facts appear in \cite[Sections~1.5.3 and~2.4]{HP}; we state their precise parameters and retain their proofs.
\begin{lemma}\label{lem:dual-short}
Let $C$ be an $[n,k,n-k+1]_q$ MDS code with $1\leq k<n$.
\begin{enumerate}[label=(\roman*)]
\item Its Euclidean dual
\[
 C^\perp=\{y\in\F_q^n:\textstyle\sum_i c_i y_i=0\text{ for every }c\in C\}
\]
is an $[n,n-k,k+1]_q$ MDS code.
\item Shortening on any $s<k$ specified coordinates produces an
\[
 [n-s,k-s,n-k+1]_q
\]
MDS code.
\end{enumerate}
\end{lemma}
\begin{proof}
For (i), choose any $k$-set $I$ and write $J=I^c$. By Lemma~\ref{lem:general-column}, a coordinate permutation and invertible row operations give
\[
 G=(I_k\mid A),\qquad
 H=(-A^{\mathsf T}\mid I_{n-k}),\qquad GH^{\mathsf T}=0.
\]
The rows of $H$ are independent, so they generate $C^\perp$. Its columns indexed by $J$ form $I_{n-k}$. Changing the generator matrix of $C^\perp$ by an invertible row operation does not change its information sets. Because $I$ was arbitrary, every $(n-k)$-set is therefore an information set for $C^\perp$. Lemma~\ref{lem:general-column} gives the claimed parameters.

For (ii), let $S$ be the shortening set, $|S|=s$, and define
\[
 U_S=\{u\in\F_q^{1\times k}:uG_S=0\}.
\]
Since $\rank G_S=s$,
\[
 \dim U_S=k-s.
\]
The map $u\mapsto(uG)_{S^c}$ is injective on $U_S$: its kernel would give $uG=0$. Thus the shortened code has dimension $k-s$. Its nonzero words have the same weight as the corresponding words of $C$, because the deleted coordinates are zero. Consequently,
\[
 n-k+1\leq d(\operatorname{Short}_S C)
 \leq(n-s)-(k-s)+1=n-k+1,
\]
and equality follows.
\end{proof}

The projective line now provides the needed numerical interval. The first inequality is the elementary linear-code form of the classical Bush bound \cite{Bush}; see also \cite[Section~2.4]{HP}. The other two inequalities follow from duality, as shown below.
\begin{lemma}\label{lem:bush}
Every nontrivial linear $[n,k]_q$ MDS code satisfies
\begin{equation}\label{eq:elementary-bounds}
 n\leq q+k-1,\qquad k\leq q-1,\qquad n\leq2q-2.
\end{equation}
\end{lemma}
\begin{proof}
Shortening on $k-2$ coordinates gives an $[n-k+2,2,n-k+1]_q$ MDS code. Its columns are nonzero and pairwise nonproportional. They represent distinct one-dimensional subspaces of $\F_q^2$, of which there are
\[
 |\PG(1,q)|=\frac{q^2-1}{q-1}=q+1.
\]
Therefore $n-k+2\leq q+1$, which is the first inequality.

Nontriviality gives $n-k\geq2$, so the same inequality applies to the MDS dual:
\[
 n\leq q+(n-k)-1
 \quad\Longrightarrow\quad k\leq q-1.
\]
Substitution yields $n\leq q+(q-1)-1=2q-2$.
\end{proof}

We shall also change scalar fields. The next observation is a direct application of the standard minor characterization of matrix rank; the coding-theoretic distance conclusion follows from \eqref{eq:distance-rank}. No novelty is claimed for this linear-algebraic fact.
\begin{lemma}\label{rem:scalar}
Let $K\subseteq L$ be fields, let $G\in K^{k\times n}$ have rank $k$, and put $C_K=\operatorname{rowspan}_K G$, $C_L=\operatorname{rowspan}_L G$. Then
\[
 \rank_K G_I=\rank_L G_I\quad\text{for every }I,
 \qquad d(C_K)=d(C_L).
\]
In particular, $C_K$ is MDS if and only if $C_L$ is MDS.
\end{lemma}
\begin{proof}
For every square submatrix $B$ of $G_I$, its determinant is an element of $K$, and
\[
 \det_K B=0\iff\det_L B=0.
\]
Since rank is the largest size of a nonzero minor, all ranks agree. Formula~\eqref{eq:distance-rank} then gives equality of minimum distances.
\end{proof}

\subsection{Cyclic ideals and full coordinate period}
Having established the general MDS criteria, we now examine how the cyclic structure restricts the coordinate period.
The ideal description of cyclic codes is standard \cite[Chapter~4]{HP}. Identify $(c_0,\ldots,c_{n-1})$ with $\sum_{i=0}^{n-1}c_iX^i$ modulo $X^n-1$. The right cyclic shift becomes multiplication by $X$. A cyclic code is therefore an ideal
\[
 C=\langle g(X)\rangle\subseteq R_n:=\F_q[X]/(X^n-1),
\]
where $g$ is a unique monic divisor of $X^n-1$.
\begin{definition}\label{def:cyclic}
The check polynomial is $h=(X^n-1)/g$. The code is \emph{irreducible} if $h$ is irreducible over $\F_q$. It is \emph{faithful} if a root of $h$ has multiplicative order $n$.
\end{definition}
The check polynomial also records the action of the cyclic shift, as the following module description shows. The map
\begin{equation}\label{eq:cyclic-module}
 \F_q[X]/(h)\longrightarrow C,\qquad
 a\bmod h\longmapsto ag\bmod(X^n-1)
\end{equation}
is an isomorphism of $\F_q[X]$-modules: its kernel is zero because
\[
 X^n-1\mid ag\iff gh\mid ag\iff h\mid a.
\]
Consequently,
\[
 \dim C=\deg h=k,
 \qquad C=\{ag:\deg a<k\},
\]
and the minimal polynomial of the right shift on $C$ is $h$.

We have not assumed $\gcd(n,q)=1$. The periodic structure underlying the next proof is classical for irreducible cyclic codes; compare \cite{Hollmann} and \cite[Section~2.5]{FHLX}. We give the argument directly, including the repeated-root case, before applying the MDS condition.
\begin{lemma}\label{lem:faithful}
An irreducible cyclic MDS code of dimension $k\geq2$ is faithful. In particular, $\gcd(n,q)=1$.
\end{lemma}
\begin{proof}
Let $\alpha$ be a root of $h$, and put $s=\ord(\alpha)$. Since $h\mid X^n-1$,
\[
 s\mid n,\qquad h\mid X^s-1,\qquad\gcd(s,p)=1.
\]
Write $n=rs$. Then
\begin{align*}
 g(X)
 &=\frac{X^n-1}{h(X)}\\
 &=\bigl(1+X^s+\cdots+X^{(r-1)s}\bigr)
     \frac{X^s-1}{h(X)}.
\end{align*}
For a codeword $ag$, $\deg a<k$, put $b=a(X)(X^s-1)/h(X)$. We have
\[
 \deg b\leq(k-1)+(s-k)=s-1,
 \qquad ag=\sum_{j=0}^{r-1}X^{js}b.
\]
Thus the coefficient vector of every word is
\[
 (\boldsymbol b,\boldsymbol b,\ldots,\boldsymbol b)
 \quad(r\text{ identical blocks of length }s).
\]
If $r>1$, any generator matrix has equal columns at positions $i$ and $i+s$. This contradicts Lemma~\ref{lem:general-column}, because $k\geq2$. Hence $r=1$, $s=n$, and $\gcd(n,p)=1$.
\end{proof}

\subsection{Derivation of the trace representation}
With full coordinate period established, we can pass from the cyclic ideal to a field-theoretic model.
The trace description is classical, rather than a new representation theorem; see \cite{Hollmann} and \cite[Section~2.5]{FHLX}. The nondegenerate trace pairing is standard finite-field theory \cite[Section~3.8]{HP}. We retain a full derivation because the inverse-root convention matters later.
The representation argument applies to every faithful irreducible cyclic code; after faithfulness is known, it does not require the MDS hypothesis. Let $C$ be such a code. Choose a root $\alpha$ of its check polynomial, and set
\[
 \beta=\alpha^{-1},\qquad L=\F_{q^k},\qquad H=\langle\beta\rangle.
\]
Then $|H|=n$, $\F_q(\beta)=L$, and Lemma~\ref{lem:orders} gives $k=\ord_n(q)$. The inverse is chosen to match the right-shift convention, not to replace $h$ silently by its reciprocal.

The trace map and its bilinear pairing are
\[
 T(x)=\Tr_{L/\F_q}(x)=\sum_{j=0}^{k-1}x^{q^j},
 \qquad B(a,x)=T(ax).
\]
They are $\F_q$-linear in each relevant argument. Also $T(x)^q=T(x)$, so $T(L)\subseteq\F_q$. The polynomial defining $T$ is nonzero and has degree $q^{k-1}<q^k$; hence it cannot vanish on all of $L$. Choose $z$ with $T(z)\neq0$. For $x\neq0$,
\begin{equation}\label{eq:trace-nondegenerate}
 B(z/x,x)=T(z)\neq0.
\end{equation}
Thus $B$ is nondegenerate, even when $p\mid k$ and $T(1)=0$.

We can now identify the original code, rather than merely construct a cyclic code with similar parameters.
\begin{proposition}\label{prop:trace-model}
With the preceding choices,
\begin{equation}\label{eq:trace}
 C=C(q,n)=\left\{
 \bigl(\Tr_{\F_{q^k}/\F_q}(a\beta^i)\bigr)_{i=0}^{n-1}:
 a\in\F_{q^k}\right\}.
\end{equation}
Every primitive choice of $\beta$ in the subgroup of order $n$ gives a permutation-equivalent code.
\end{proposition}
\begin{proof}
Let $\Theta:L\to\F_q^n$ denote the displayed map. If $\Theta(a)=0$, then
\[
 T(a\beta^i)=0\quad(0\leq i<k).
\]
The powers $1,\beta,\ldots,\beta^{k-1}$ form a basis of $L$, so $T(ax)=0$ for all $x\in L$. By \eqref{eq:trace-nondegenerate}, $a=0$. Therefore $\Theta$ is injective and $\dim\Theta(L)=k$.

To identify the image with $C$, it remains to compare the shift actions. If $S$ is the right shift, then, with indices modulo $n$,
\[
 (S\Theta(a))_i=T(a\beta^{i-1})=T(\alpha a\beta^i),
 \qquad S\Theta(a)=\Theta(\alpha a).
\]
Thus $\Theta(L)$ is cyclic, and the shift on it has minimal polynomial $h$, the minimal polynomial of $\alpha$. By \eqref{eq:cyclic-module}, a cyclic ideal has shift minimal polynomial equal to its check polynomial. The unique ideal with check polynomial $h$ is $C$; hence $\Theta(L)=C$.

Every element of order $n$ in $L^\times$ is $\beta^u$ with $\gcd(u,n)=1$. For such $u$,
\[
 \bigl(T(a(\beta^u)^i)\bigr)_i
 =\bigl(T(a\beta^{ui})\bigr)_i,
\]
which is the coordinate permutation $i\mapsto ui$. This proves the final assertion.
\end{proof}

The remaining step is to express the MDS condition in terms of the subgroup $H$. The next lemma is the classical MDS column criterion \cite[Section~2.4]{HP} transported through the trace pairing; it is not a separate novelty claim. Its explicit formulation for the subgroup drives the rest of the paper.
\begin{lemma}\label{lem:column}
The code in \eqref{eq:trace} is MDS if and only if every $k$ distinct elements of $H$ are independent over $\F_q$. In that case every set of at most $k$ distinct elements of $H$ is independent.
\end{lemma}
\begin{proof}
Choose a basis $a_1,\ldots,a_k$ of $L$ and define
\[
 \Lambda:L\longrightarrow\F_q^k,
 \qquad x\longmapsto(T(a_1x),\ldots,T(a_kx))^{\mathsf T}.
\]
Nondegeneracy of $B$ gives $\ker\Lambda=0$, so $\Lambda$ is an isomorphism. The generator matrix of the trace code with rows $\Theta(a_j)$ is
\[
 G=\bigl(\Lambda(1),\Lambda(\beta),\ldots,\Lambda(\beta^{n-1})\bigr).
\]
For any distinct $i_1,\ldots,i_s$ and any $c_1,\ldots,c_s\in\F_q$,
\[
 \sum_{j=1}^s c_j\Lambda(\beta^{i_j})=0
 \iff
 \Lambda\!\left(\sum_{j=1}^s c_j\beta^{i_j}\right)=0
 \iff
 \sum_{j=1}^s c_j\beta^{i_j}=0.
\]
The assertion follows from Lemma~\ref{lem:general-column}.
\end{proof}

\section{Subgroup rigidity and length congruences}\label{sec:rigidity}
The trace model identifies each coordinate with a vector in $L=\F_{q^k}$. A subgroup of the coordinate orbit is especially useful because its span is a subfield. The MDS condition will force that subfield to be all of $L$.

\subsection{Rigidity of every nonidentity subgroup}
We begin with the order constraint from which the arithmetic restrictions in this section follow.
\begin{theorem}\label{thm:rigidity}
Let $C$ be a nontrivial irreducible cyclic $[n,k]_q$ MDS code. Then
\begin{equation}\label{eq:rigidity}
 \ord_d(q)=k\qquad(d\mid n,\ d>1).
\end{equation}
In particular, every prime $\ell\mid n$ satisfies
\begin{equation}\label{eq:prime-order}
 \ord_\ell(q)=k,\qquad k\mid\ell-1,
\end{equation}
and
\begin{gather}
 n\text{ is odd},\qquad
 \gcd(n,q-1)=\gcd(n,k)=1,\label{eq:basic-consequences}\\
 n=mk+1\quad\text{for an integer }m\geq2.\label{eq:dimension-congruence}
\end{gather}
\end{theorem}
\begin{proof}
Let $H=\langle\beta\rangle$ be the column subgroup. For $d>1$ dividing $n$, put
\[
 \gamma=\beta^{n/d},\qquad H_d=\langle\gamma\rangle,
 \qquad r=\ord_d(q).
\]
Then $|H_d|=d$, and Lemma~\ref{lem:orders} gives
\[
 \F_q(\gamma)=\F_{q^r}\subseteq\F_{q^k},\qquad r\mid k.
\]
In fact $\langle H_d\rangle_{\F_q}=\F_q(\gamma)$: powers of $\gamma$ lie in $H_d$, and they span its minimal-polynomial quotient. The column criterion rules out a proper subfield here. If $r<k$, then
\[
 r\leq\varphi(d)\leq d-1,
 \qquad r+1\leq\min\{d,k\}.
\]
Choose $r+1$ distinct elements of $H_d$. They lie in the $r$-dimensional space $\F_{q^r}$ and are dependent over $\F_q$, contradicting Lemma~\ref{lem:column}. Hence $r=k$, proving \eqref{eq:rigidity}.

Now take $d=\ell$ prime. Lagrange's theorem in $(\Z/\ell\Z)^\times$ gives
\[
 k=\ord_\ell(q)\mid\ell-1,
 \qquad\ell\geq k+1\geq3.
\]
All prime factors of $n$ are therefore odd and larger than $k$. This proves that $n$ is odd and $\gcd(n,k)=1$. If a prime $\ell$ divided both $n$ and $q-1$, then
\[
 k=\ord_\ell(q)=1,
\]
contrary to nontriviality. Thus $\gcd(n,q-1)=1$.

These restrictions on the individual prime factors also impose a congruence on the length itself. Writing $n=\prod_{\ell\mid n}\ell^{v_\ell(n)}$, we obtain
\[
 n\equiv\prod_{\ell\mid n}1^{v_\ell(n)}=1\pmod k.
\]
Since $n\geq k+2$, the representation $n=mk+1$ has $m\neq0,1$, and hence $m\geq2$.
\end{proof}

The conclusion is stronger than irreducibility alone. Irreducibility determines the order modulo $n$; the MDS condition determines it modulo every nonidentity divisor of $n$. Two useful numerical consequences follow immediately.
\begin{corollary}\label{cor:factor-count}
Under the hypotheses of Theorem~\ref{thm:rigidity},
\begin{equation}\label{eq:factor-bounds}
 n\geq(k+1)^{\Omega(n)},\qquad
 q\geq(k+1)^{\Omega(n)}-k+1.
\end{equation}
If $n$ is composite, then $q\geq k^2+k+2$. If $k$ is odd, every prime divisor of $n$ is at least $2k+1$.
\end{corollary}
\begin{proof}
Theorem~\ref{thm:rigidity} and Lemma~\ref{lem:bush} give
\[
 n=\prod_{\ell\mid n}\ell^{v_\ell(n)}
 \geq\prod_{\ell\mid n}(k+1)^{v_\ell(n)}
 =(k+1)^{\Omega(n)},
 \qquad q\geq n-k+1.
\]
For composite $n$, $\Omega(n)\geq2$, so
\[
 q\geq(k+1)^2-k+1=k^2+k+2.
\]
The prime-divisor congruence gives a further improvement when $k$ is odd. If $k$ is odd and $\ell=ak+1$ is an odd prime, then $ak=\ell-1$ is even. Thus $a$ is even and $a\geq2$, giving $\ell\geq2k+1$.
\end{proof}

\subsection{Congruence bounds and the dimension-two endpoint}
The preceding corollary is most effective when $n$ has several factors. Without factoring $n$, one can instead intersect $n\equiv1\pmod k$ with the interval $n\leq q+k-1$.
\begin{corollary}\label{cor:residue}
For a nontrivial irreducible cyclic $[n,k]_q$ MDS code,
\begin{align}
 k\mid q-1&\quad\Longrightarrow\quad n\leq q-k,\label{eq:residue-one}\\
 k\mid q,\quad k\geq3&\quad\Longrightarrow\quad n\leq q-k+1.\label{eq:residue-zero}
\end{align}
In particular, the MDS conjecture holds whenever $q\equiv0$ or $1\pmod k$.
\end{corollary}
\begin{proof}
If $q=uk+1$, then $n\equiv1\pmod k$ and
\[
 n\leq q+k-1=(u+1)k
 \quad\Longrightarrow\quad n\leq uk+1=q.
\]
Equality would give $n=q$, contrary to $\gcd(n,q)=1$. The next admissible integer is $q-k$.

If $q=uk$ and $k\geq3$, then
\[
 n\equiv1\pmod k,\qquad n\leq(u+1)k-1
 \quad\Longrightarrow\quad n\leq uk+1=q+1.
\]
At $n=q+1>2$,
\[
 q\equiv-1\pmod n,
 \qquad\ord_n(q)=2\neq k.
\]
Hence $n\leq q-k+1$. When $k=2$, Lemma~\ref{lem:bush} already gives $n\leq q+1$.
\end{proof}

The following parameter choices illustrate how these congruences can improve the general length bound.
\begin{example}
For $q=2^{10}=1024$ and $k=32$, \eqref{eq:residue-zero} gives $n\leq993$. For $q=2^{12}=4096$ and $k=65$, the divisibility $65\mid4095$ gives $n\leq4031$. These are necessary bounds, not existence assertions at the endpoints.
\end{example}

Dimension two deserves separate treatment: the column condition then tests only pairs, so the intersection with the scalar field gives a complete classification. The criterion below is the irreducible recurrence-code case of Shi, Zhang, and Sol\'e \cite[Theorems~1 and~5]{SZS}. In their notation, the projective period is $e=n/\gcd(n,q-1)$ and the repetition factor is $K=n/e=\gcd(n,q-1)$; the MDS case is $K=1$. We retain the direct subgroup proof and derive the endpoint consequence in the same proposition.
\begin{proposition}\label{prop:dimension-two}
A faithful irreducible cyclic code of dimension two is MDS if and only if $\gcd(n,q-1)=1$. Its possible lengths are exactly the odd divisors $n>1$ of $q+1$. Thus $n=q+1$ occurs in this dimension exactly when $q$ is even.

For a nontrivial irreducible cyclic MDS code of dimension $k\geq3$,
\begin{equation}\label{eq:strict-endpoint}
 n\leq q+1\quad\Longrightarrow\quad n\leq q-1.
\end{equation}
\end{proposition}
\begin{proof}
Let $H\leq\F_{q^2}^\times$ have order $n$. Two distinct elements $x,y\in H$ are proportional over $\F_q$ precisely when
\[
 xy^{-1}\in\F_q^\times\setminus\{1\}.
\]
Consequently,
\begin{align*}
 C(q,n)\text{ is MDS}
 &\iff H\cap\F_q^\times=\{1\}\\
 &\iff\gcd(n,q-1)=1.
\end{align*}
The last equivalence follows because both groups are subgroups of the cyclic group $\F_{q^2}^\times$, so their intersection has order $\gcd(n,q-1)$.

We now translate the intersection criterion into a description of the possible lengths. Since $n\mid(q-1)(q+1)$ and $\gcd(n,q-1)=1$, we have $n\mid q+1$. This $n$ is odd: if $q$ is even then $q+1$ is odd; if $q$ is odd then $q-1$ is even and coprimality excludes $2\mid n$.

Conversely, let $n>1$ be an odd divisor of $q+1$. Then
\[
 \gcd(n,q-1)\mid\gcd(q+1,q-1)\mid2,
 \qquad\gcd(n,q-1)=1,
\]
and $q\equiv-1\pmod n$ has order exactly two, since $n>2$. Thus the subgroup of order $n$ gives a faithful irreducible dimension-two MDS code. The length $q+1$ is odd exactly when $q$ is even.

Finally, for $k\geq3$, the equality $n=q+1$ would force $k=\ord_n(q)=2$, while $n=q$ contradicts Lemma~\ref{lem:faithful}. Removing these two integers proves \eqref{eq:strict-endpoint}.
\end{proof}

\subsection{Normal elements in the coordinate subgroup}
Subgroup rigidity also determines the lengths of Frobenius orbits. The column criterion then upgrades each nonidentity orbit to a basis. Normal elements are classical \cite[Chapter~2]{LN}; the conclusion below is a necessary-condition deduction from subgroup rigidity, not an invocation of the normal basis theorem.
\begin{definition}
An element $x\in\F_{q^k}$ is \emph{normal over $\F_q$} if
\[
 \{x,x^q,\ldots,x^{q^{k-1}}\}
\]
is an $\F_q$-basis of $\F_{q^k}$.
\end{definition}
We apply this definition to each nonidentity element of $H$ and then use the resulting trace values to determine the weight of a particular codeword.
\begin{proposition}\label{prop:normal}
For a nontrivial irreducible cyclic MDS code with column subgroup $H$, every $x\in H\setminus\{1\}$ is normal over $\F_q$. In particular,
\[
 \Tr_{\F_{q^k}/\F_q}(x)\neq0\qquad(x\in H\setminus\{1\}).
\]
The word in \eqref{eq:trace} with parameter $a=1$ has weight
\begin{equation}\label{eq:trace-word-weight}
 \wt\Theta(1)=
 \begin{cases}
 n,&p\nmid k,\\
 n-1,&p\mid k.
 \end{cases}
\end{equation}
\end{proposition}
\begin{proof}
For $x\in H\setminus\{1\}$, put $d=\ord(x)$. Theorem~\ref{thm:rigidity} gives $\ord_d(q)=k$. If $0\leq a<b<k$ and $x^{q^a}=x^{q^b}$, then
\[
 d\mid q^a(q^{b-a}-1)
 \quad\Longrightarrow\quad q^{b-a}\equiv1\pmod d,
\]
a contradiction. Thus the $k$ Frobenius conjugates of $x$ are distinct elements of $H$. Lemma~\ref{lem:column} makes them independent, and hence a basis. The equation
\[
 \Tr(x)=\sum_{j=0}^{k-1}x^{q^j}=0
\]
would be a nonzero relation with every coefficient equal to $1$, so it is impossible.

This nonvanishing result determines the weight of the distinguished word. Every coordinate of $\Theta(1)$ except the coordinate indexed by $0$ is therefore nonzero. At that remaining coordinate,
\[
 \Tr(1)=\sum_{j=0}^{k-1}1=k\cdot1,
\]
which vanishes exactly when $p\mid k$. This proves \eqref{eq:trace-word-weight}.
\end{proof}
Normality is only a necessary condition. It checks particular $k$-subsets, whereas the MDS condition requires independence of \emph{every} $k$-subset. The distinction will remain important after the arithmetic reductions.

\section{The permutation automorphism group}\label{sec:standard}
The preceding restrictions are sufficient for the length arguments developed later. Here we give a correct local extension-field reduction and a conditional consequence of the classification statement in \cite{FHLX}. Another purpose of this section is to prove Corollary~\ref{cor:not-two-transitive}, which shows, under that classification assumption, that the two-transitive method of \cite{DFV} is inapplicable to nontrivial irreducible cyclic MDS codes. The distinction is essential: a local correction does not establish the global classification. This entire section may be omitted from either length proof.

\subsection{From permutations to linear stabilizers}
The permutation automorphism group is
\[
 \PAut(C)=\{\pi\in\operatorname{Sym}(\Z/n\Z):\pi(C)=C\},
\]
where $(\pi c)_{\pi(i)}=c_i$. Put $L=\F_{q^k}$ and
\[
 G_H=\{T\in\operatorname{GL}_{\F_q}(L):T(H)=H\}.
\]
The identification of code permutations with the linear stabilizer of the defining subgroup is standard; see \cite{Hollmann} and \cite[Section~2.5]{FHLX}. Since $H$ spans $L$, restriction to $H$ is faithful. The generator-column identification of Lemma~\ref{lem:column} gives
\begin{equation}\label{eq:stabilizer-identification}
 G_H\cong\PAut(C).
\end{equation}
Indeed, a linear map carrying each column to its permuted column preserves the generator row space. Conversely, if a coordinate permutation preserves that row space, two generator matrices for the same code differ by an invertible row operation; that operation induces a unique linear map on $L$ carrying the columns accordingly.

With this identification in place, two elements of $G_H$ are always available:
\[
 \sigma(x)=\beta x,\qquad\psi(x)=x^q.
\]
On coordinate indices they induce
\begin{equation}\label{eq:standard-generators}
 \sigma:i\mapsto i+1,\qquad \psi:i\mapsto qi.
\end{equation}
They satisfy
\[
 \sigma^n=\psi^k=1,\qquad
 \psi\sigma\psi^{-1}=\sigma^q.
\]
Every element of $\langle\sigma,\psi\rangle$ has the form $\sigma^a\psi^j$, with $0\leq a<n$ and $0\leq j<k$. This representation is unique: applying an equality of two such maps to $1$ first determines $a$, and then the order of Frobenius on $L$ determines $j$. Hence
\[
 |\langle\sigma,\psi\rangle|=nk,
 \qquad\langle\sigma,\psi\rangle\cong C_n\rtimes_q C_k.
\]
Following \cite[Section~2.3]{FHLX}, $C$ is \emph{standard} if this is its full permutation group, and \emph{nonstandard} otherwise.

\subsection{Defects and the classification constructions}
The arithmetic obstruction needed for the classification can be stated independently of any code.
\begin{definition}
For $n>1$ and $\gcd(n,q)=1$, the pair $(n,q)$ is \emph{defective} if there is a divisor $d>1$ of $n$ such that
\[
 \ord_d(q)<\ord_n(q).
\]
\end{definition}
By Theorem~\ref{thm:rigidity}, a nontrivial irreducible cyclic MDS code never has a defective pair. We next explain why a defect cannot disappear along the construction sequences in \cite{FHLX}.

\begin{lemma}\label{lem:defect}
Defectiveness persists under Constructions~3.7--3.10 of \cite{FHLX}. A nonidentity extension in Construction~3.7 and every tensor construction in Construction~3.10 introduce a defect. The equally spaced family in Example~3.5 of that paper is defective.
\end{lemma}
\begin{proof}
For an input pair $(n,q)$, write $m=\ord_n(q)$. When a defect is already present, let $d\mid n$, $d>1$, and $r=\ord_d(q)<m$. Lemma~\ref{lem:orders} gives $r\mid m$.

\smallskip\noindent\emph{Extension, Construction~3.7.}
The output has length $nu$ over $\F_q$, where
\[
 u~ \bigg| ~\frac{q-1}{\gcd(n,q-1)},\qquad \ord_{nu}(q)=m.
\]
The same $d$ remains a divisor and has order $r<m$. If $u>1$, any prime $\ell\mid u$ satisfies
\[
 \ell\mid q-1,\qquad\ell\mid nu,\qquad\ord_\ell(q)=1<m,
\]
so a defect is introduced even without an input defect.

Here $m>1$ because the constructions start with a nonstandard faithful code. To justify this restriction, a faithful dimension-one code has a generator row whose entries are the distinct elements of a subgroup $H\leq\F_q^\times$. A preserving permutation multiplies that row by some $c\in\F_q^\times$ with $cH=H$. Since $1\in H$,
\[
 cH=H\iff c\in H.
\]
These permutations are exactly the cyclic shifts, so such a code is standard.

\smallskip\noindent\emph{Lifting, Construction~3.8.}
The output alphabet is $\F_{q^t}$ with $\gcd(t,m)=1$. Since $r\mid m$,
\[
 \ord_n(q^t)=\frac m{\gcd(m,t)}=m,
 \qquad
 \ord_d(q^t)=\frac r{\gcd(r,t)}=r<m.
\]

\smallskip\noindent\emph{Special descent, Construction~3.9.}
Write the output alphabet as $\F_q$, the input alphabet as $\F_{q^s}$, and set
\[
 M=\ord_n(q),\quad s\mid M,\quad R=\ord_d(q)\mid M.
\]
If $d$ witnesses an input defect, then
\[
 \frac R{\gcd(R,s)}<\frac M{\gcd(M,s)}.
\]
Equality $R=M$ would make these fractions equal. Thus $R<M$, and the same divisor witnesses an output defect.

\smallskip\noindent\emph{Tensor construction, Construction~3.10.}
The output has length $n'=\lcm(n,s)$, with
\[
 t=\ord_s(q)>1,\qquad\gcd(m,t)=1.
\]
By \eqref{eq:order-lcm},
\[
 \ord_{n'}(q)=\lcm(m,t)=mt>m=\ord_n(q).
\]
Thus the old length $n$ is itself a defective divisor of the new length. Any preexisting defect also persists.

Finally, in the equally spaced family,
\[
 n=an',\qquad a\geq2,\qquad n'>1,\qquad
 \ord_n(q)=a\ord_{n'}(q).
\]
The proper divisor $n'$ witnesses the defect.
\end{proof}

The construction argument uses field-linear stabilizers, while the classification also encounters extension-field \emph{semilinear} groups. Before applying it, we record the precise distinction and the reduction valid here.

\subsection{The extension-field reduction}
A map $g$ is $E$-semilinear if there is $\tau_g\in\operatorname{Gal}(E/\F_q)$ such that
\[
 g(ex)=\tau_g(e)g(x)\qquad(e\in E,\ x\in L).
\]
It is $E$-linear exactly when $\tau_g=1$. Normalizing a field of endomorphisms yields semilinearity, not automatically linearity. The following reduction identifies the linear kernel while retaining the Galois action on the normalized field.
\begin{proposition}\label{prop:semilinear-repair}
Let $H=\langle\beta\rangle\leq L^\times$, where $L=\F_{q^k}$, $|H|=n$, $\ord_n(q)=k$, and $\gcd(n,k)=1$. Let $G=G_H$. Suppose that $G$ normalizes a field
\[
 E\cong\F_{q^b}\subseteq\operatorname{End}_{\F_q}(L),
 \qquad\F_q\subseteq E,\qquad b\mid k.
\]
Then $E$ is the canonical subfield of the multiplication field $L$. Moreover, with $G_E=G\cap\operatorname{GL}_E(L)$,
\begin{gather}
 G=G_E\langle\psi\rangle,\qquad
 G_E=\operatorname{Stab}_{\operatorname{GL}_{k/b}(q^b)}(H),\label{eq:linear-kernel}\\
 1\longrightarrow G_E\longrightarrow G
 \longrightarrow\operatorname{Gal}(E/\F_q)\longrightarrow1.\label{eq:semilinear-exact}
\end{gather}
If $G$ is nonstandard over $\F_q$, then $G_E$ is nonstandard over $E$.
\end{proposition}
\begin{proof}
Conjugation gives a homomorphism
\[
 \rho:G\longrightarrow\operatorname{Aut}_{\F_q}(E),\qquad
 \rho(g)(e)=geg^{-1}.
\]
The automorphism $\rho(\sigma)$ has order dividing $n=\ord(\sigma)$ and $b=|\operatorname{Aut}_{\F_q}(E)|$. Since $b\mid k$ and $\gcd(n,k)=1$,
\[
 \ord\rho(\sigma)\mid\gcd(n,b)=1.
\]
Therefore $E$ centralizes multiplication by $\beta$.

We compute that centralizer directly. If $T\sigma=\sigma T$ and $c=T(1)$, then
\[
 T(\beta^i)=T(\sigma^i1)=\sigma^iT(1)=\beta^ic
 \qquad(0\leq i<k).
\]
The powers $1,\beta,\ldots,\beta^{k-1}$ form a basis, so $T(x)=cx$ for every $x\in L$. Conversely, every multiplication map commutes with $\sigma$. Hence
\[
 C_{\operatorname{End}_{\F_q}(L)}(\sigma)
 =\{x\mapsto cx:c\in L\}\cong L.
\]
There is a unique subfield of $L$ of order $q^b$, which proves the assertion about $E$.

Having identified $E$ inside the multiplication field, we now determine the induced Galois action. For $e\in E$, direct conjugation gives
\[
 \psi(e\,\psi^{-1}(x))=e^q x.
\]
Thus $\rho(\psi)$ is the $q$-Frobenius on $E$ and generates its Galois group. In particular $\rho$ is onto, and
\[
 \ker\rho=G\cap\operatorname{GL}_E(L)=G_E.
\]
For any $g\in G$, choose $j$ with $\rho(g)=\rho(\psi)^j$. Then $g\psi^{-j}\in G_E$, proving \eqref{eq:linear-kernel} and \eqref{eq:semilinear-exact}.

Finally,
\[
 \ord_n(q^b)=\frac{k}{\gcd(k,b)}=\frac kb.
\]
Thus $H$ is the faithful column subgroup over $E$ with dimension $k/b$, whose standard stabilizer is $\langle\sigma,\psi^b\rangle$. If $G_E$ were this group, then
\[
 G=G_E\langle\psi\rangle
   =\langle\sigma,\psi^b,\psi\rangle
   =\langle\sigma,\psi\rangle,
\]
contrary to nonstandardness of $G$.
\end{proof}

This local reduction clarifies one step of the classification argument; we next specify the assumption needed for the global conclusion.
\begin{remark}[Scope of the classification input]\label{rem:classification-audit}
The coprimality $\gcd(n,k)=1$ required above follows from Theorem~\ref{thm:rigidity}. The standardness deduction below assumes the classification statement of \cite[Theorem~1.1]{FHLX}. We do not regard Proposition~\ref{prop:semilinear-repair} as a proof of that global classification.

The printed argument in \cite[Lemma~5.4(2)]{FHLX} does not justify replacing extension-field semilinearity by linearity: its field automorphism is taken independent of the group element, whereas semilinearity supplies an automorphism $\tau_g$ for each $g$. These notions cannot in general be identified. In the canonical-field branch with $E=\F_{q^b}$ and $b>1$, discussed in \cite[Proposition~5.8]{FHLX}, one has
\[
 \psi(ex)=e^q\psi(x)\neq e\psi(x)
 \quad\text{when }e^q\neq e.
\]
The Frobenius image in $\operatorname{Gal}(E/\F_q)$ has order $b$, and the Frobenius generator in the standard group over $E$ is $\psi^b$. Proposition~\ref{prop:semilinear-repair} retains this quotient and supplies a valid local reduction under the MDS coprimality hypothesis. It does not verify the other classification branches, so the downstream standardness assertion remains conditional on the classification statement. The issue is an unjustified step in the printed proof; no counterexample to its main classification statement is asserted.

Likewise, the nonstandardness assertion in the tensor construction can be justified without the printed group-order comparison. If $\gcd(m,t)=1$, then
\[
 \F_{q^m}\otimes_{\F_q}\F_{q^t}\cong\F_{q^{mt}}.
\]
Extend a nonstandard input map $f$ uniquely $\F_{q^t}$-linearly. It stabilizes the product of the input subgroup and the subgroup in $\F_{q^t}^{\times}$. If the extension were $x\mapsto\lambda x^{q^j}$, then $\lambda=f(1)$ would lie in the input subgroup; restriction to $\F_{q^m}$ would make $f$ standard, a contradiction. These observations clarify the reductions used here; they do not constitute a new verification of every argument in the classification. Both length theorems are independent of it.
\end{remark}

\subsection{Exclusion of all nonstandard possibilities}
We can now combine the arithmetic obstruction with the classification. The surviving construction sequences have a particularly simple form: they preserve the length and cannot change the dimension of a prime-field starting code.
\begin{theorem}\label{thm:standard}
Assume the classification statement of \cite[Theorem~1.1]{FHLX}: every nonstandard nondegenerate irreducible cyclic code is obtained from its Examples~3.1--3.5 by its Constructions~3.7--3.10. Then every nontrivial irreducible cyclic MDS code is standard. Thus
\begin{equation}\label{eq:full-standard-group}
 \PAut(C)=\langle\sigma,\psi\rangle\cong C_n\rtimes_q C_k,
 \qquad\psi\sigma\psi^{-1}=\sigma^q.
\end{equation}
\end{theorem}
\begin{proof}
Suppose that $C$ is nonstandard. Lemma~\ref{lem:faithful} gives the coprimality and full-period hypotheses of the classification. The classification obtains $C$ from one of its five initial families by Constructions~3.7--3.10.

\smallskip\noindent\emph{Step 1: remove every defective path.}
A defective intermediate pair would remain defective by Lemma~\ref{lem:defect}, contradicting Theorem~\ref{thm:rigidity} at the output. Therefore the equally spaced initial family, every nonidentity extension, and every tensor construction are excluded.

To determine which initial families can survive this exclusion, first consider the full multiplicative-orbit family, whose parameters are
\[
 n=Q^a-1,\qquad\ord_n(Q)=a\geq2.
\]
For $Q>2$, the divisor $Q-1>1$ has order one and is defective. Similarly, in the extension ovoid family,
\[
 n=(Q-1)(Q^2+1),\qquad\ord_n(Q)=4,
\]
and $Q-1$ is defective whenever $Q>2$. The remaining possible initial pairs therefore all lie over prime fields:
\[
 \begin{array}{c|c|c}
 \text{initial family}&\text{alphabet}&(n,k_0)\\\hline
 \text{repetition dual}&\F_p&(n,n-1)\\
 \text{binary simplex}&\F_2&(2^{k_0}-1,k_0),\ k_0\geq3\\
 \text{binary ovoid}&\F_2&(5,4)\\
 \text{binary Golay dual}&\F_2&(23,11)\\
 \text{ternary Golay dual}&\F_3&(11,5).
 \end{array}
\]
The range $k_0\geq3$ in the binary simplex row reflects the exclusion of $(a,Q)=(2,2)$ from the nonstandard initial family.

\smallskip\noindent\emph{Step 2: show that surviving paths are scalar extensions.}
Only lifting and special descent remain. They preserve $n$ and the characteristic $p$. Write their successive fields as $\F_{p^{a_j}}$ and put
\[
 k_0=\ord_n(p),\qquad
 k_j=\ord_n(p^{a_j})=\frac{k_0}{\gcd(k_0,a_j)}\leq k_0.
\]
Lifting preserves dimension. In special descent from $\F_{q^s}$ to $\F_q$, with $s\mid M=\ord_n(q)$, the input and output dimensions are $M/s$ and $M$; hence descent cannot decrease dimension. Starting at dimension $k_0$, this gives
\[
 k_j\geq k_0\quad\text{and}\quad k_j\leq k_0,
 \qquad k_j=k_0,\quad\gcd(k_0,a_j)=1.
\]
Thus the subgroup of exponents generated by $p^{a_j}$ modulo $n$ equals that generated by $p$. The check-root orbit is unchanged, up to a primitive multiplier. By Proposition~\ref{prop:trace-model}, the output is permutation-equivalent to a scalar extension of an initial prime-field code.

\smallskip\noindent\emph{Step 3: exclude the surviving initial codes.}
The repetition dual and binary ovoid rows have dimension $n-1$ and remain trivial under scalar extension. The other rows are nontrivial, and Lemma~\ref{lem:bush} over their initial fields would require
\[
 \begin{array}{c|c}
 \text{alphabet}&\text{necessary MDS inequality}\\\hline
 \F_2&k_0\leq1,\\
 \F_3&k_0\leq2.
 \end{array}
\]
The simplex and Golay dimensions in the table violate these inequalities. By Lemma~\ref{rem:scalar}, scalar extension cannot turn a non-MDS code into an MDS code. Every possible initial family is therefore excluded.

This contradiction proves standardness. The normal form and conjugation relation established above give \eqref{eq:full-standard-group}.
\end{proof}

With the standard group identified, we can compare its point stabilizers with the requirement of two-transitivity.
\begin{corollary}\label{cor:not-two-transitive}
Under the classification assumption in Theorem~\ref{thm:standard}, the permutation group of a nontrivial irreducible cyclic MDS code is not two-transitive.
\end{corollary}
\begin{proof}
Its stabilizer of the coordinate $0$ is
\[
 \PAut(C)_0=\{i\mapsto q^ji:0\leq j<k\},
 \qquad|\PAut(C)_0|=k.
\]
Two-transitivity would require this stabilizer to act transitively on the remaining $n-1$ coordinates, and hence $n-1\leq k$. But Theorem~\ref{thm:rigidity} gives
\[
 n-1=mk\geq2k>k,
\]
a contradiction.
\end{proof}
Thus, under the stated classification assumption, the two-transitive theorem of \cite{DFV} does not apply to the remaining codes. Irrespective of that assumption, cyclic transitivity alone does not justify its use. We now return to length arguments that use no automorphism classification.

\section{Exact field descent}\label{sec:descent}
The column subgroup constrains the extension degree, but the same code may already have a generator matrix over a smaller alphabet. The coefficient-field argument below is an application of the standard Frobenius-orbit description of minimal polynomials \cite[Section~3.7 and Chapter~4]{HP}; we do not claim that description as new. We give its exact fixed-coordinate formulation and then derive the MDS obstruction.

\subsection{The coefficient field of the cyclic ideal}
For $K\subseteq\F_q$, a code $C\leq\F_q^n$ is \emph{defined over $K$ in the fixed coordinates} if
\[
 C=\operatorname{span}_{\F_q}C_0
 \quad\text{for some }C_0\leq K^n.
\]
No coordinate scalings or permutations are allowed in this definition. Thus it concerns the actual cyclic realization, not merely an equivalent code. The next theorem identifies the smallest such field by a multiplicative-order calculation and then applies the MDS bound over that field.
\begin{theorem}\label{thm:descent}
Let $q=p^f$, $n>1$, and $\gcd(n,p)=1$. Put
\[
 e=\ord_n(p),\qquad a=\gcd(e,f),\qquad
 k=\ord_n(q)=\frac ea.
\]
Then every faithful irreducible cyclic $[n,k]_q$ code is the scalar extension of an irreducible cyclic $[n,k]_{p^a}$ code. Moreover,
\begin{equation}\label{eq:definition-field}
 C\text{ is defined over }\F_{p^b},\ b\mid f
 \quad\Longleftrightarrow\quad a\mid b.
\end{equation}
In particular, $\F_{p^a}$ is its smallest field of definition in the fixed cyclic coordinates.

If $C$ is nontrivial and MDS and $n>q+1$, then
\begin{equation}\label{eq:descent-necessary}
 a=f,\qquad \ord_n(p)=fk,\qquad fk\mid\varphi(n).
\end{equation}
\end{theorem}
\begin{proof}
Choose a root $\alpha$ of the check polynomial $h$, so that $\ord(\alpha)=n$. Its root set is
\[
 \mathcal O_f=\{\alpha^{p^{fj}}:0\leq j<k\}.
\]
Inside the cyclic group $\langle p\rangle\leq(\Z/n\Z)^\times$ of order $e$,
\begin{equation}\label{eq:orbit-subgroups}
 \langle p^f\rangle=\langle p^{\gcd(e,f)}\rangle
 =\langle p^a\rangle.
\end{equation}
For completeness, $a$ is an integer combination of $e$ and $f$, so $p^a$ is a power of $p^f$ modulo $n$; conversely, $a\mid f$ makes $p^f$ a power of $p^a$. Thus the two Frobenius orbits agree.

It follows that
\[
 h(X)=\prod_{j=0}^{k-1}(X-\alpha^{p^{aj}})
 \in\F_{p^a}[X],
 \qquad\ord_n(p^a)=\frac ea=k.
\]
The polynomial remains irreducible over $\F_{p^a}$, since the displayed orbit has size $k$. Polynomial division by the monic $h$ gives
\[
 g(X)=\frac{X^n-1}{h(X)}\in\F_{p^a}[X].
\]
The rows corresponding to $g,Xg,\ldots,X^{k-1}g$ therefore generate the required descended code and its scalar extension.

We next determine all possible subfields. Let $b\mid f$. The coefficients of $h$ lie in $\F_{p^b}$ precisely when coefficient Frobenius permutes its roots:
\begin{align*}
 h\in\F_{p^b}[X]
 &\iff\mathcal O_f^{p^b}=\mathcal O_f\\
 &\iff p^b\langle p^f\rangle=\langle p^f\rangle\\
 &\iff p^b\in\langle p^a\rangle\\
 &\iff a\mid b.
\end{align*}
The last equivalence uses $a\mid e$: membership means $b\equiv aj\pmod e$ for some $j$, which is equivalent to divisibility by $a$.

To pass from polynomial coefficients to the code itself, suppose that $C=\operatorname{span}_{\F_q}C_0$ with $C_0\leq\F_{p^b}^n$. Choose a full-row-rank generator $G_0$ of $C_0$. Scalar extension preserves its rank, so it has $k$ rows. Choose pivot columns $I$ with $(G_0)_I$ invertible. If $c=uG_0\in C\cap\F_{p^b}^n$, then
\[
 u=c_I(G_0)_I^{-1}\in\F_{p^b}^{1\times k}.
\]
Hence
\begin{equation}\label{eq:intersection-descent}
 C_0=C\cap\F_{p^b}^n.
\end{equation}
The right side is stable under cyclic shift, because both $C$ and $\F_{p^b}^n$ are stable. Therefore $C_0$ is cyclic. Its monic generator $g_0$ also generates its scalar extension, so uniqueness gives $g_0=g$. Then $h=(X^n-1)/g$ also has coefficients in $\F_{p^b}$. The converse follows by taking the rows $g,Xg,\ldots,X^{k-1}g$ over that subfield. This proves \eqref{eq:definition-field}.

Having identified the field of definition, we turn to its consequence for length. Suppose $C$ is nontrivial and MDS. If $a<f$, then $a\mid f$ implies $f\geq2a$. Lemmas~\ref{rem:scalar} and~\ref{lem:bush}, applied over $\F_{p^a}$, give
\begin{align*}
 n&\leq2p^a-2\leq2\sqrt q-2<q+1,\\
 q+1-(2\sqrt q-2)&=(\sqrt q-1)^2+2>0.
\end{align*}
This contradicts $n>q+1$. Thus $a=f$, and $k=e/a$ yields $e=fk$. Lagrange's theorem gives $e\mid\varphi(n)$, proving \eqref{eq:descent-necessary}.
\end{proof}

The theorem isolates a useful obstruction: a code that genuinely descends to a proper subfield is much too short to violate the target bound over the larger field. The remaining case has maximal prime-field order.

\subsection{A dimension restriction and an example}
Combining the maximal-order condition with the elementary upper bound on $n$ gives a restriction on the dimension.
\begin{corollary}\label{cor:high-dim}
For $q=p^f$ with $f\geq2$, every nontrivial irreducible cyclic MDS code satisfies the conjectured bound whenever
\[
 (f-1)k>q-2.
\]
Any prime-length counterexample must also satisfy $fk\mid n-1$.
\end{corollary}
\begin{proof}
A counterexample would give, successively,
\[
 fk=\ord_n(p)\leq\varphi(n)\leq n-1\leq q+k-2.
\]
Subtracting $k$ gives $(f-1)k\leq q-2$, a contradiction. If $n$ is prime, then $\varphi(n)=n-1$, and the divisibility in \eqref{eq:descent-necessary} becomes $fk\mid n-1$.
\end{proof}

The distinction between the orders of $q$ and $p$ modulo $n$ is already visible in the following small parameter set.
\begin{example}
For $(q,n,k)=(9,11,5)$,
\[
 \ord_{11}(9)=5,
 \qquad\ord_{11}(3)=5,
 \qquad a=\gcd(5,2)=1.
\]
The first order condition is consistent with subgroup rigidity, since $11$ is prime. Nevertheless, the code descends to $\F_3$. A nontrivial MDS code over that field would have
\[
 n\leq2\cdot3-2=4,
\]
which excludes $n=11$ without a distance search.
\end{example}
This example shows why the subsequent sieve must test both $\ord_n(q)$ and the stronger prime-field condition, rather than only the first.

\section{Cyclotomic reduction and finite dimensions}\label{sec:finite}
The elementary MDS bound leaves only a short interval above $q+1$. The substitution $t=n-q$ turns that interval into a bounded set of integers depending on $k$ alone. Subgroup rigidity then assigns every prime-power divisor of $n$ to one cyclotomic factor.

\subsection{A finite arithmetic reduction}
We first record this reduction in a form that retains the full prime-power divisibility of the length.
\begin{theorem}\label{thm:cyclotomic}
Suppose a nontrivial irreducible cyclic $[n,k]_q$ MDS code has $n>q+1$. With $t=n-q$,
\begin{gather}
 2\leq t\leq k-1,\label{eq:finite-interval}\\
 n~\Big|~|\Phi_k(-t)|,\qquad
 \ord_\ell(-t)=k\quad(\ell\mid n\text{ prime}).\label{eq:finite}
\end{gather}
Thus the length conjecture in every fixed dimension reduces to finitely many explicitly specified pairs $(q,n)$.
\end{theorem}
\begin{proof}
By Lemma~\ref{lem:bush},
\[
 q+2\leq n\leq q+k-1,
\]
which is \eqref{eq:finite-interval}. Since $q\equiv-t\pmod n$ and $\ord_n(q)=k$,
\[
 n\mid(-t)^k-1.
\]
For every prime $\ell\mid n$, Theorem~\ref{thm:rigidity} gives
\[
 \ord_\ell(-t)=\ord_\ell(q)=k,
 \qquad k\mid\ell-1.
\]
In particular $\ell\nmid t$ and $\ell\nmid k$.

To locate each prime-power contribution, evaluate \eqref{eq:cyclotomic-product} at $-t$:
\[
 (-t)^k-1=\prod_{d\mid k}\Phi_d(-t).
\]
If $d<k$ divides $k$ and $\ell\mid\Phi_d(-t)$, then $\Phi_d(X)\mid X^d-1$ gives
\[
 (-t)^d\equiv1\pmod\ell,
\]
contrary to the order $k$. Therefore all the $\ell$-adic valuation is in the factor $\Phi_k(-t)$:
\begin{align*}
 v_\ell(n)
 &\leq v_\ell\bigl((-t)^k-1\bigr)\\
 &=\sum_{d\mid k}v_\ell\bigl(\Phi_d(-t)\bigr)
 =v_\ell\bigl(\Phi_k(-t)\bigr).
\end{align*}
This proves divisibility by the entire prime-power factor $\ell^{v_\ell(n)}$, not only by $\ell$. Multiplication over $\ell\mid n$ proves \eqref{eq:finite}.

Finally, for fixed $k$, there are only $k-2$ possible values of $t$, and $\Phi_k(-t)\neq0$ for $t\geq2$. Each of these nonzero integers has finitely many positive divisors. For each retained divisor $n$, the alphabet is forced to be $q=n-t$.
\end{proof}

There is also a useful uniqueness statement before any cyclotomic factorization is performed.
\begin{corollary}\label{cor:one-length}
Write $q=uk+r$, where $0\leq r<k$. A counterexample is possible only if $2\leq r\leq k-1$, in which case
\begin{equation}\label{eq:unique-length}
 n=(u+1)k+1,\qquad t=k+1-r.
\end{equation}
The displayed integer must still satisfy all order, parity, and field-descent conditions.
\end{corollary}
\begin{proof}
The interval $[q+2,q+k-1]$ has fewer than $k$ integers, so contains at most one integer congruent to $1$ modulo $k$. If $n=vk+1$, then
\[
 t=n-q=(v-u)k+1-r.
\]
The inequalities $2\leq t\leq k-1$ force $v-u=1$ and $2\leq r\leq k-1$. Substitution gives \eqref{eq:unique-length}.
\end{proof}

The cyclotomic theorem is useful in any fixed dimension. In the range needed here, a second reduction avoids large cyclotomic factorizations altogether. It uses only the established geometric ranges in Lemma~\ref{lem:geometric}, whose proof in Section~\ref{sec:composite} is independent of this section.

\subsection{A bounded arithmetic reduction}
Applying these geometric ranges together with exact field descent yields the following list of possible exceptions.
\begin{lemma}\label{lem:twelve}
If a nontrivial irreducible cyclic $[n,k]_q$ MDS code has $2\leq k\leq30$ and $n>q+1$, then $(q,n,k)$ is one of the triples in Table~\ref{tab:residual}.
\end{lemma}
\begin{proof}
The case $k=2$ is excluded by Proposition~\ref{prop:dimension-two}. For $k=3,4,5$, the cyclotomic values are respectively
\[
 \{3\},\qquad \{5,10\},\qquad \{11,61,205\},
\]
where $205=5\cdot41$. Every prime factor of $n$ must be $1$ modulo $k$, and $n\geq2k+1$. The first two sets give no admissible length. For $k=5$, the only possibilities are
\[
 (q,n)=(9,11),\quad(58,61),\quad(37,41).
\]
The middle alphabet is not a prime power; the last is excluded by Ball's prime-field theorem \cite{Ball}; the first descends to $\F_3$ by Theorem~\ref{thm:descent} and is excluded by Lemma~\ref{lem:bush}.

We may therefore assume $6\leq k\leq30$. If $q>3136$, then
\[
 q>3125,\qquad q>4(30-2)^2\geq4(k-2)^2,
\]
contradicting Lemma~\ref{lem:geometric}(i). Hence $q\leq3136$. Write $q=p^f$. For a prime alphabet, Ball's bound \cite[Theorem~1.7]{Ball} gives $n\leq q+k+1-\min\{k,p\}=q+1$, since Lemma~\ref{lem:bush} gives $k\leq q-1=p-1$. For $f>1$, \cite[Theorem~5.2]{BallDeBeule} gives $n\leq q+1$ whenever $k\leq2p-2$. Thus we may assume
\[
 f\geq2,\qquad k>2p-2.
\]
These restrictions leave the following complete set of alphabets to be considered:
\[
\begin{array}{c|l}
p&\{p^f:2\leq f,\ p^f\leq3136\}\\\hline
2&4,8,16,32,64,128,256,512,1024,2048\\
3&9,27,81,243,729,2187\\
5&25,125,625,3125\\
7&49,343,2401\\
11&121,1331\\
13&169,2197.
\end{array}
\]
We now impose the restrictions on the length and the relevant multiplicative orders. For each displayed $q=p^f$ and each integer
\[
 \max\{6,2p-1\}\leq k\leq\min\{30,q-1\},
\]
write $q=uk+r$. Corollary~\ref{cor:one-length} leaves at most the single integer $n=(u+1)k+1$, and only when $2\leq r\leq k-1$. The remaining necessary tests are exactly
\begin{gather}
 n\geq2k+1,\qquad \gcd(n,p)=1,\label{eq:bounded-tests-a}\\
 \ord_\ell(q)=k\quad(\ell\mid n\text{ prime}),
 \qquad \ord_n(p)=fk.\label{eq:bounded-tests-b}
\end{gather}
Their simultaneous solutions in this displayed finite domain are the twelve rows of Table~\ref{tab:residual}. This last calculation requires only trial division of integers $n\leq3165$ and modular powers: for any coprime $a,m$, an asserted order $e$ is checked by
\[
 a^e\equiv1\pmod m,\qquad
 a^{e/r}\not\equiv1\pmod m\quad(r\mid e\text{ prime}).
\]
Every possible alphabet and dimension is included in the ranges above, and every discarded condition was necessary. This proves the claimed exhaustion.
\end{proof}

\begin{table}[htbp]
\centering
\caption{The complete residual parameter set from Lemma~\ref{lem:twelve}. All displayed lengths are prime.}\label{tab:residual}
\begin{tabular}{@{}rrrrrr@{}}
\toprule
$q$&$n$&$k$&$p$&$f$&$n-q$\\\midrule
27&31&10&3&3&4\\
64&67&11&2&6&3\\
49&53&13&7&2&4\\
125&127&14&5&3&2\\
25&37&18&5&2&12\\
25&43&21&5&2&18\\
81&89&22&3&4&8\\
25&47&23&5&2&22\\
125&139&23&5&3&14\\
81&101&25&3&4&20\\
243&251&25&3&5&8\\
49&61&30&7&2&12\\
\bottomrule
\end{tabular}
\end{table}

\subsection{A uniform polynomial obstruction}
A nonzero multiple of the minimal polynomial of $\beta$ gives a dependence among powers of $\beta$. The next argument arranges sufficiently many zero coefficients in a multiple of small degree. It is an application of ordinary polynomial multiplication and the column criterion; the underlying polynomial description of cyclic-code duality is standard \cite[Chapter~4]{HP}.

To formulate the obstruction simultaneously for all primitive roots, we first work in a polynomial quotient ring. Fix a prime $n\neq p$ and $q=p^f$ with $k=\ord_n(q)$. In
\[
 R_{p,n}=\F_p[Y]/(Y^n-1)
\]
define $c_0(Y),\ldots,c_k(Y)$ by
\begin{equation}\label{eq:orbit-polynomial}
 \prod_{h=0}^{k-1}\bigl(X-Y^{q^h\bmod n}\bigr)
    =\sum_{j=0}^k c_j(Y)X^j\quad\text{in }R_{p,n}[X].
\end{equation}
Set $c_j=0$ for $j\notin\{0,\ldots,k\}$. Every residue class in $R_{p,n}$ has a unique representative of degree less than $n$, which is used whenever polynomial divisibility is stated. Form the coefficient matrix
\begin{equation}\label{eq:toeplitz}
 T_{q,n}(Y)=\bigl(c_{r-j}(Y)\bigr)_{
             \substack{0\leq r\leq k+3\\0\leq j\leq3}}
\end{equation}
and the product of its square minors of order four:
\begin{equation}\label{eq:minor-product}
 \Delta_{q,n}(Y)=
 \prod_{0\leq r_0<r_1<r_2<r_3\leq k+3}
       \det\bigl(c_{r_i-j}(Y)\bigr)_{0\leq i,j\leq3}
       \quad\text{in }R_{p,n}.
\end{equation}
Since $\Phi_n(Y)$ divides $Y^n-1$, reduction modulo $\Phi_n$ is well-defined on this ring. No choice of an extension-field basis occurs in these definitions.

Vanishing of this product at a primitive root forces a suitable four-row submatrix to be singular, as the next lemma explains.
\begin{lemma}\label{lem:low-multiple}
Suppose $k+3<n$ and
\begin{equation}\label{eq:polynomial-obstruction}
 \Phi_n(Y)\mid\Delta_{q,n}(Y)\quad\text{in }\F_p[Y].
\end{equation}
Then no faithful irreducible cyclic code with parameters $(q,n,k)$ is MDS.
\end{lemma}
\begin{proof}
Let $\beta$ be any primitive $n$th root in an algebraic closure of $\F_p$. The polynomial
\[
 m_\beta(X)=\prod_{h=0}^{k-1}(X-\beta^{q^h})
          =\sum_{j=0}^k c_j(\beta)X^j
\]
is the minimal polynomial of $\beta$ over $\F_q$: its $k$ roots are distinct, and $q$-Frobenius permutes them transitively. In particular $c_j(\beta)\in\F_q$ and $c_k(\beta)=1$.

Equation~\eqref{eq:polynomial-obstruction} gives $\Delta_{q,n}(\beta)=0$. A product in a field is zero only if one factor is zero. Thus some four-row submatrix of $T_{q,n}(\beta)$ is singular over $\F_q$, and it has a nonzero kernel vector $(b_0,b_1,b_2,b_3)^{\mathsf T}\in\F_q^4$. Put
\[
 b(X)=b_0+b_1X+b_2X^2+b_3X^3\neq0,
 \qquad a(X)=m_\beta(X)b(X).
\]
The coefficient of $X^r$ in $a$ is
\[
 [X^r]a=\sum_{j=0}^3 c_{r-j}(\beta)b_j.
\]
Four distinct coefficients, at the chosen rows, are zero. Since $a\neq0$ and $\deg a\leq k+3<n$, it follows that
\[
 \wt(a)\leq(k+4)-4=k,\qquad a(\beta)=0.
\]
The nonzero coefficients of $a$ therefore give a nonzero $\F_q$-linear dependence among at most $k$ distinct elements of $\langle\beta\rangle$. Lemma~\ref{lem:column} rules out the MDS property. The argument holds for every primitive choice of $\beta$.
\end{proof}

It remains to verify the divisibility hypothesis for the twelve residual triples. The following identities do this over the prime field.
\begin{lemma}\label{lem:finite-identity}
For every row of Table~\ref{tab:residual}, the coefficients defined by
\eqref{eq:orbit-polynomial} satisfy
\begin{equation}\label{eq:finite-identity}
 \Delta_{q,n}(Y)\equiv0\pmod{1+Y+\cdots+Y^{n-1}}
 \quad\text{over }\F_p.
\end{equation}
\end{lemma}
\begin{proof}
Write
\[
 S_n(Y)=1+Y+\cdots+Y^{n-1},\qquad
 D_{r_0,r_1,r_2,r_3}(Y)
 =\det\bigl(c_{r_i-j}(Y)\bigr)_{0\leq i,j\leq3}.
\]
All polynomial equalities in this proof are in
$R_{p,n}=\F_p[Y]/(Y^n-1)$.
The following are the required explicit identities. In each row,
$\mathcal R$ is a set of increasing four-tuples, and
\begin{equation}\label{eq:small-minor-product}
 \prod_{R\in\mathcal R}D_R(Y)=\eta S_n(Y).
\end{equation}
\begin{center}\small
\begin{tabular}{@{}rrrlr@{}}
\toprule
$q$&$n$&$k$&$\mathcal R$&$\eta\in\F_p$\\\midrule
27&31&10&$\{(0,1,4,7)\}$&0\\
64&67&11&$\{(1,3,8,10)\}$&1\\
49&53&13&$\{(0,1,3,15),(0,1,6,12)\}$&0\\
125&127&14&$\{(0,1,3,9),(0,2,4,11),(0,3,9,10)\}$&0\\
25&37&18&$\{(0,1,2,5)\}$&4\\
25&43&21&$\{(0,1,2,5)\}$&0\\
81&89&22&$\{(0,1,3,13)\}$&2\\
25&47&23&$\{(0,1,3,10)\}$&1\\
125&139&23&$\{(0,1,4,21),(0,1,5,10)\}$&0\\
81&101&25&$\{(0,1,3,27)\}$&0\\
243&251&25&$\{(0,1,6,14),(0,1,9,11)\}$&2\\
49&61&30&$\{(0,1,3,6)\}$&0\\
\bottomrule
\end{tabular}
\end{center}
We give the coefficient rules and an expansion to specify the arithmetic in
\eqref{eq:small-minor-product} completely. Put $e_h=q^h\bmod n$,
$c^{(0)}_0=1$, and $c^{(0)}_j=0$ for $j\neq0$. Then
\begin{equation}\label{eq:coefficient-recursion}
 c^{(h+1)}_j(Y)=c^{(h)}_{j-1}(Y)-Y^{e_h}c^{(h)}_j(Y),
 \qquad 0\leq h<k,
\end{equation}
with $c_j=c^{(k)}_j$. Equivalently,
\[
 c_j(Y)=(-1)^{k-j}
 \sum_{\substack{A\subseteq\{0,\ldots,k-1\}\\|A|=k-j}}
 Y^{\sum_{h\in A}e_h\bmod n}.
\]
To multiply the resulting coefficient polynomials, use the circular coefficient multiplication rule
\begin{equation}\label{eq:circular-convolution}
 [Y^a](uv)=\sum_{b=0}^{n-1}[Y^b]u\,[Y^{a-b\bmod n}]v.
\end{equation}
Since $n$ is prime, $q\not\equiv1\pmod n$, and $\ord_n(q)=k$,
\[
 \sum_{h=0}^{k-1}q^h\equiv0\pmod n,
 \qquad c_0=(-1)^k.
\]
The boundary rows of the coefficient matrix simplify most displayed determinants at once to
\begin{equation}\label{eq:minor-boundary-reduction}
 D_{0,1,u,v}
   =c_0^2\bigl(c_{u-2}c_{v-3}-c_{u-3}c_{v-2}\bigr).
\end{equation}
The remaining three determinants have their ordinary four-by-four
expansion. These rules yield exactly the identities
\eqref{eq:small-minor-product}.

For example, consider $(q,n,k)=(27,31,10)$ over $\F_3$. Set
\[
 H=\{1,2,4,8,15,16,23,27,29,30\},\qquad
 E_i(Y)=\sum_{h\in3^iH\bmod31}Y^h\quad(i\in\Z/3\Z).
\]
Direct multiplication of these sums gives
\[
 E_i^2=1+E_{i+1}+2E_{i+2},\qquad
 E_iE_{i+1}=E_i+2E_{i+1}+E_{i+2}.
\]
The coefficient recursion gives
\begin{align*}
 c_1&=2E_0,&
 c_2&=2+E_0+2E_1+E_2,\\
 c_4&=1+2E_0+2E_1+E_2,&
 c_5&=1+2E_1.
\end{align*}
Consequently,
\[
 c_2c_4=2E_1+E_2=c_1c_5,
 \qquad D_{0,1,4,7}=0.
\]
This is the first identity in \eqref{eq:small-minor-product}; the other
rows use the same coefficient rules, with the one to three determinants
explicitly indicated.

Finally, every $D_R$ appearing in \eqref{eq:small-minor-product} is a
factor of $\Delta_{q,n}$. Reducing that identity modulo $S_n$ therefore
gives $\Delta_{q,n}\equiv0\pmod{S_n}$, as required.
\end{proof}

For completeness, we record Rabin's classical irreducibility criterion \cite{Rabin}, which is useful when working with factors over the prime field. The multiplier argument above does not require a particular irreducible-polynomial model.
\begin{lemma}\label{lem:rabin}
Let $F\in\F_p[X]$ be monic of degree $D\geq1$. It is irreducible if and only if
\begin{gather}
 X^{p^D}\equiv X\pmod F,\label{eq:rabin-main}\\
 \gcd\bigl(F,X^{p^{D/r}}-X\bigr)=1
 \quad\text{for every prime }r\mid D.\label{eq:rabin-gcd}
\end{gather}
\end{lemma}
\begin{proof}
An irreducible polynomial of degree $d$ divides $X^{p^s}-X$ exactly when its roots lie in $\F_{p^s}$, equivalently when $d\mid s$. Also
\[
 (X^{p^s}-X)'=-1,
\]
so $X^{p^s}-X$ is square-free. If $F$ is irreducible of degree $D$, then $D\mid D$ gives \eqref{eq:rabin-main}, while $D\nmid D/r$ gives each coprimality condition.

Conversely, \eqref{eq:rabin-main} makes $F$ square-free and every irreducible factor degree $d$ a divisor of $D$. If some such factor had $d<D$, choose a prime $r\mid D/d$. Then $d\mid D/r$, so that factor would divide $X^{p^{D/r}}-X$, contradicting \eqref{eq:rabin-gcd}. Thus every irreducible factor has degree $D$. Since $\deg F=D$ and $F$ is monic, it has exactly one such factor and is irreducible.
\end{proof}

\subsection{The finite-dimensional conclusion}
The bounded parameter list and the polynomial obstruction now combine to settle all dimensions in the stated range.
\begin{theorem}\label{thm:thirty}
Every nontrivial irreducible cyclic MDS code with $2\leq k\leq30$ satisfies $n\leq q+1$. For $3\leq k\leq30$, the stronger inequality $n\leq q-1$ holds.
\end{theorem}
\begin{proof}
A counterexample would have one of the parameters in Lemma~\ref{lem:twelve}. For each of those parameters, $n\geq2k+1>k+3$, and Lemma~\ref{lem:finite-identity} supplies \eqref{eq:polynomial-obstruction}. Lemma~\ref{lem:low-multiple} excludes every row. This proves $n\leq q+1$. For $k\geq3$, Proposition~\ref{prop:dimension-two} removes the two endpoints $q$ and $q+1$.
\end{proof}

\section{The MDS conjecture for every composite length}\label{sec:composite}
The finite-dimensional result does not by itself control unbounded dimension. At composite lengths, however, subgroup rigidity forces prime factors to be large relative to $k$. This puts most parameters inside established geometric ranges. Exact field descent eliminates the few factorizations that remain.

\subsection{Geometric inputs and their precise use}
For clarity, a \emph{normal rational curve} in $\PG(r-1,q)$ is a projective image of
\begin{equation}\label{eq:nrc}
 \mathcal N_{r-1}=
 \{[1:a:\cdots:a^{r-1}]:a\in\F_q\}
 \cup\{[0:\cdots:0:1]\}.
\end{equation}
It has $q+1$ rational points; in the plane it is a nonsingular conic. We will use the following numerical consequences of the cited geometric theorems, as recorded in \cite[Section~19]{BallLavrauw}.
\begin{lemma}\label{lem:geometric}
Let $C$ be a nontrivial linear $[n,k]_q$ MDS code with $k\geq6$.
\begin{enumerate}[label=(\roman*)]
\item If $q>3125$ and $k<\tfrac12\sqrt q+2$, then $n\leq q+1$.
\item If $q=p^{2a+1}$ with $p$ odd, $a\geq1$, and
\begin{equation}\label{eq:voloch}
 k<\frac14\sqrt{pq}-\frac{29}{16}p+4,
\end{equation}
then $n\leq q+1$.
\end{enumerate}
\end{lemma}
\begin{proof}[Derivation from the cited results]
For odd $q$, the Hirschfeld--Korchm\'aros range \cite{HK}, in the precise form recorded in \cite[Section~19, p.~28]{BallLavrauw}, is
\[
 q\geq529,\quad q\notin\{729,3125\},\quad
 k\leq\tfrac12\sqrt q+2.
\]
For even $q$, the Storme--Thas range \cite{ST}, as recorded in the same table, is
\[
 k<\tfrac12\sqrt q+\tfrac{15}{4}.
\]
The stronger common restrictions in (i) lie inside these respective ranges. Since $k\geq6$ and, for even $q>3125$, $k<q-1$, neither exceptional MDS dimension needs to be retained. This proves (i).

For (ii), we use the Voloch range stated in \cite[Section~19]{BallLavrauw}, originating in \cite{Voloch}. Here is the projection mechanism and its constant. The planar bound recorded as \cite[Theorem~53]{BallLavrauw} states conic containment for every planar arc of size greater than
\[
 M_0(q)=q-\frac14\sqrt{pq}+\frac{29}{16}p-1.
\]
We use the deliberately weaker threshold $M(q)=M_0(q)+2$, namely
\[
 M(q)=q-\frac14\sqrt{pq}+\frac{29}{16}p+1.
\]
This weakening gives exactly the strict range \eqref{eq:voloch}; the $+1$ here is not the constant in the quoted planar theorem. Suppose, for contradiction, that $n\geq q+2$, and select a $(q+2)$-subset $\mathcal A$ of the generator arc.

To apply the planar bound, we project from subsets of the arc. For a subset $S\subseteq\mathcal A$ of size $s\leq k-3$, let $U_S$ be the vector span of representatives of its points. Projection means taking the remaining vectors modulo $U_S$:
\[
 \mathcal A_S\subseteq\PG(\F_q^k/U_S),\qquad
 \dim U_S=s,\qquad|\mathcal A_S|=q+2-s.
\]
These projected vectors are an arc. Indeed, a dependence among at most $k-s$ quotient vectors would lift to a dependence among at most $k$ original representatives, including those of $S$. This also proves that no remaining point becomes zero or coincides projectively with another.

At $s=k-3$ the quotient is three-dimensional and
\[
 |\mathcal A_S|=q+5-k>M(q),
\]
where the strict inequality is exactly \eqref{eq:voloch}. Thus every such planar projection lies on a conic.

To lift this conic containment back to the original space, use the two-point projection theorem \cite[Theorem~24]{BallLavrauw}: an arc in $\PG(r-1,q)$, $r\geq4$ and $q\geq r+1$, lies on a normal rational curve if its projections from two distinct arc points do so. Starting with $s=k-3$ and decreasing $s$ one at a time, apply this theorem to $\mathcal A_S$. Its two projections are $\mathcal A_{S\cup\{x\}}$ and $\mathcal A_{S\cup\{y\}}$, already covered by the induction hypothesis. The field-size condition holds because Lemma~\ref{lem:bush} gives $q\geq k+1\geq(k-s)+1$.

At $s=0$ we conclude that $\mathcal A$ lies on a normal rational curve. This is impossible because
\[
 |\mathcal A|=q+2>q+1=|\mathcal N_{k-1}|.
\]
The deep planar bounds and the two-point projection theorem are external inputs; the quotient argument above makes their application explicit.
\end{proof}

\subsection{The small composite-length enumeration}
Before applying these ranges uniformly, we remove a bounded set of small dimensions. The following integer argument is independent of the polynomial identity in Lemma~\ref{lem:finite-identity}.
\begin{lemma}\label{lem:small-composite}
Every nontrivial irreducible cyclic MDS code of composite length and dimension at most thirty satisfies $n\leq q+1$.
\end{lemma}
\begin{proof}
Dimension two is Proposition~\ref{prop:dimension-two}. Suppose $3\leq k\leq30$, $n$ is composite, and $n>q+1$.

For $k\in\{3,4,5\}$, the evaluations in Theorem~\ref{thm:cyclotomic} are
\[
 \begin{array}{c|c|l}
 k&\Phi_k(X)&\{\Phi_k(-t):2\leq t\leq k-1\}\\\hline
 3&X^2+X+1&3\\
 4&X^2+1&5,\ 10=2\cdot5\\
 5&X^4+X^3+X^2+X+1&11,\ 61,\ 205=5\cdot41.
 \end{array}
\]
Every prime factor of an admissible $n$ must be $1$ modulo $k$. For $k=3$ no displayed factor qualifies. For $k=4$ only the single factor $5$ qualifies; for $k=5$ each evaluation has at most one qualifying prime factor, with exponent one. Thus none has an admissible composite divisor.

Assume now $6\leq k\leq30$. If $q>3136$, then
\[
 q>3125,\qquad
 q>3136=4(30-2)^2\geq4(k-2)^2,
\]
so $k<\tfrac12\sqrt q+2$, contradicting Lemma~\ref{lem:geometric}(i). Hence $q\leq3136$.

Write $q=p^f$. Ball's theorem \cite[Theorem~1.7]{Ball}, using $k\leq q-1$, and the nonprime-alphabet theorem \cite[Theorem~5.2]{BallDeBeule} allow us to assume
\[
 f\geq2,\qquad k>2p-2.
\]
Since $k\leq30$, this gives $p<16$, and hence
\[
 \begin{array}{c|rrrrrr}
 p&2&3&5&7&11&13\\\hline
 \max\{f:p^f\leq3136\}&11&7&5&4&3&3.
 \end{array}
\]
Every power with $2\leq f$ up to these maxima is tested. For each one, enumerate $6\leq k\leq30$, $k>2p-2$, and $q+2\leq n\leq q+k-1$. Retain only
\begin{equation}\label{eq:small-sieve}
 n\geq2k+1,\quad \Omega(n)\geq2,\quad\gcd(n,q)=1,
 \quad \ell\equiv1\pmod k\quad(\ell\mid n\text{ prime}).
\end{equation}
These elementary tests leave exactly Table~\ref{tab:small-composite}. Trial division suffices, since every tested $n$ satisfies $n\leq3136+29=3165$.

\begin{table}[htbp]
\centering\small
\caption{All small composite candidates after \eqref{eq:small-sieve}; the final order differs from $k$ in every row.}\label{tab:small-composite}
\begin{tabular}{@{}rrrlrr@{}}
\toprule
$q$&$n$&$k$&factorization of $n$&$d\mid n$&$\ord_d(q)$\\\midrule
128&133&6&$7\cdot19$&7&3\\
243&247&6&$13\cdot19$&13&3\\
256&259&6&$7\cdot37$&7&3\\
1024&1027&6&$13\cdot79$&79&39\\
2187&2191&6&$7\cdot313$&313&39\\
3125&3131&10&$31\cdot101$&31&3\\
2187&2197&12&$13^3$&13&3\\
3125&3133&12&$13\cdot241$&13&4\\
2048&2059&14&$29\cdot71$&29&28\\
2401&2413&18&$19\cdot127$&19&3\\
512&529&22&$23^2$&23&11\\
2187&2209&23&$47^2$&2209&1081\\
\bottomrule
\end{tabular}
\end{table}

We now apply the full order condition to the candidates in the table. Each row has a divisor $d>1$ with $\ord_d(q)\neq k$, contradicting Theorem~\ref{thm:rigidity}. The sieve intentionally uses weaker conditions than full irreducibility; orders greater than $k$ are therefore valid exclusion witnesses as well. In the last row it is enough to check
\[
 2187^{23}\equiv706\not\equiv1\pmod{2209}.
\]
This excludes all small-dimensional composite counterexamples.
\end{proof}

\subsection{Reduction to three factorizations}
We can now let the dimension vary without a bound. Large prime factors increase the length, and the elementary MDS bound transfers that increase to $q$, placing the parameters in the geometric range.
\begin{proposition}\label{prop:three-shapes}
If a nontrivial irreducible cyclic MDS code has composite length $n>q+1$, then $k\geq31$ is even, $\ell=k+1$ is prime, and
\begin{equation}\label{eq:three-shapes}
 n=\ell^2,\qquad n=\ell(2k+1),\qquad
 \text{or }n=\ell(3k+1),
\end{equation}
where the other displayed factor is also prime.
\end{proposition}
\begin{proof}
Lemma~\ref{lem:small-composite} gives $k\geq31$. Every prime divisor of $n$ is $1$ modulo $k$ by Theorem~\ref{thm:rigidity}.

We first bound the number of prime factors, counted with multiplicity. Suppose $\Omega(n)\geq3$. Corollary~\ref{cor:factor-count} gives
\[
 q\geq(k+1)^3-k+1\geq32^3-31+1=32738>3125.
\]
Also
\begin{align*}
 (k+1)^3-k+1-4(k-2)^2
 &=k^3-k^2+18k-14\\
 &=k^2(k-1)+18k-14>0.
\end{align*}
Thus $q>4(k-2)^2$, or $k<\tfrac12\sqrt q+2$. Lemma~\ref{lem:geometric}(i) gives a contradiction. Since $n$ is composite, the only remaining possibility is $\Omega(n)=2$.

With exactly two prime factors remaining, we next bound the smaller one. Write $n=rs$, with $r,s$ prime, allowing $r=s$. If both $r,s\geq2k+1$, then
\[
 q\geq n-k+1\geq(2k+1)^2-k+1
 =4k^2+3k+2\geq3939>3125,
\]
and
\[
 (4k^2+3k+2)-4(k-2)^2=19k-14>0.
\]
The same geometric contradiction follows. Therefore one factor is smaller than $2k+1$. The only positive multiple of $k$ plus one in the permitted interval $[k+1,2k]$ is $k+1$. Thus one factor is the prime $\ell=k+1$. Since $\ell\geq3$ is odd, $k$ is even.

It remains to bound the other prime factor. If this factor is at least $4k+1$, then
\[
 q\geq(k+1)(4k+1)-k+1=4k^2+4k+2>3125,
\]
while
\[
 (4k^2+4k+2)-4(k-2)^2=20k-14>0.
\]
Again Lemma~\ref{lem:geometric}(i) is a contradiction. The remaining factor is $1$ modulo $k$ and lies in $[k+1,4k]$, so it is one of
\[
 k+1,\qquad2k+1,\qquad3k+1.
\]
This gives \eqref{eq:three-shapes}.
\end{proof}

The geometric reduction has left only two prime factors with tightly controlled orders. Exact descent now makes the extension exponent $f$ restrictive enough to eliminate all three forms.
\begin{proposition}\label{prop:eliminate}
None of the factorizations in Proposition~\ref{prop:three-shapes} can occur for a counterexample.
\end{proposition}
\begin{proof}
Write $q=p^f$. Ball's prime-field theorem excludes $f=1$, so $f>1$. Theorem~\ref{thm:descent} gives
\begin{equation}\label{eq:full-prime-order}
 \ord_n(p)=fk.
\end{equation}
For each divisor $m>1$ of $n$, put $d_m=\ord_m(p)$. Subgroup rigidity and \eqref{eq:order-power} give
\begin{equation}\label{eq:divisor-prime-orders}
 k=\ord_m(p^f)=\frac{d_m}{\gcd(d_m,f)}.
\end{equation}
For the prime $\ell=k+1$, $d_\ell\mid\ell-1=k$. Hence
\[
 k=\frac{d_\ell}{\gcd(d_\ell,f)}\leq d_\ell\leq k,
\]
so equality holds throughout:
\begin{equation}\label{eq:first-prime}
 d_\ell=k,\qquad\gcd(k,f)=1.
\end{equation}
Since $k$ is even, $f$ is odd. These common order restrictions now allow us to treat the three factorizations in turn.

\smallskip\noindent\emph{Case 1: $n=\ell^2$.}
The reduction map
\[
 (\Z/\ell^2\Z)^\times\longrightarrow(\Z/\ell\Z)^\times
\]
has kernel of order $\ell$. The order of $p$ modulo $\ell^2$ is a multiple of its order $k$ modulo $\ell$, and the quotient of these orders divides the kernel order. Therefore
\[
 d_{\ell^2}\in\{k,\ell k\}.
\]
Together with \eqref{eq:full-prime-order} and $f>1$, this forces $f=\ell$. Thus
\[
 q=p^\ell\geq2^\ell>\ell^2=n,
\]
contradicting $n>q+1$. For the strict inequality, $2^5>5^2$, and for every integer $u\geq3$,
\[
 2u^2-(u+1)^2=u^2-2u-1>0.
\]
Induction proves $2^u>u^2$ for all $u\geq5$, which includes $\ell=k+1$.

\smallskip\noindent\emph{Case 2: $n=\ell r$, $r=2k+1$ prime.}
Equation~\eqref{eq:divisor-prime-orders} shows $k\mid d_r$, while $d_r\mid r-1=2k$. Hence $d_r\in\{k,2k\}$. By \eqref{eq:first-prime} and oddness of $f$,
\[
 \gcd(2k,f)=1.
\]
The possibility $d_r=2k$ would imply
\[
 \ord_r(p^f)=\frac{2k}{\gcd(2k,f)}=2k\neq k.
\]
Thus $d_r=k$, and the two primes are distinct. Formula~\eqref{eq:order-lcm} yields
\[
 \ord_n(p)=\lcm(d_\ell,d_r)=\lcm(k,k)=k,
\]
contradicting $\ord_n(p)=fk$ with $f>1$.

\smallskip\noindent\emph{Case 3: $n=\ell r$, $r=3k+1$ prime.}
Now $k\mid d_r\mid3k$, so $d_r\in\{k,3k\}$. Therefore
\[
 fk=\ord_n(p)=\lcm(k,d_r)\in\{k,3k\},
\]
and $f>1$ forces $f=3$. In particular $q=p^3$.

We can now compare the size of the remaining length with the geometric bound over a cubic extension. Put $t=n-q$. From $t\leq k-1$,
\begin{equation}\label{eq:cubic-size}
 p^3=(k+1)(3k+1)-t
 \geq3k^2+3k+2>3k^2.
\end{equation}
Since $k\geq31$,
\[
 p^3\geq3\cdot31^2+3\cdot31+2=2978>13^3.
\]
The prime $p$ is therefore at least $17$, in particular odd. For such $p$,
\begin{gather*}
 k<\frac{p^{3/2}}{\sqrt3},\qquad
 \frac{p^{3/2}}{\sqrt3}\leq\frac{p^2}{7}
 \iff49\leq3p,\\
 \frac{p^2}{4}-\frac{29p}{16}+4-\frac{p^2}{7}
 =\frac{p(12p-203)}{112}+4>0.
\end{gather*}
Combining these gives
\[
 k<\frac{p^2}{4}-\frac{29p}{16}+4
   =\frac14\sqrt{p\,q}-\frac{29}{16}p+4.
\]
The odd nonsquare field $q=p^3$ consequently satisfies Lemma~\ref{lem:geometric}(ii), a final contradiction.
\end{proof}

\subsection{The composite-length theorem and the main theorem}
The reductions now fit together without any additional case analysis.
\begin{theorem}\label{thm:composite}
Every nontrivial irreducible cyclic MDS code of composite length satisfies $n\leq q+1$. If $k\geq3$, then $n\leq q-1$.
\end{theorem}
\begin{proof}
If $k\leq30$, apply Lemma~\ref{lem:small-composite}. Otherwise a counterexample would, by Proposition~\ref{prop:three-shapes}, have one of the three forms in \eqref{eq:three-shapes}; Proposition~\ref{prop:eliminate} excludes every form. Thus $n\leq q+1$ in all dimensions. For $k\geq3$, Proposition~\ref{prop:dimension-two} removes the endpoints $q$ and $q+1$.
\end{proof}
Together with the finite-dimensional result, this yields the main theorem stated in the introduction.
\begin{proof}[Proof of Theorem~\ref{thm:main}]
The composite-length assertion is Theorem~\ref{thm:composite}; the arbitrary-length assertion for $k\leq30$ is Theorem~\ref{thm:thirty}. In either case, \eqref{eq:strict-endpoint} gives $n\leq q-1$ when $k\geq3$.
\end{proof}

The following remark makes the independence of the two proof chains explicit.
\begin{remark}
The two proof chains are
\[
 \begin{gathered}
 \text{trace model and subgroup rigidity}
 +\text{exact descent}
 +\text{geometric bounds}\\
 +\text{small composite enumeration}
 \quad\Longrightarrow\quad\text{Theorem~\ref{thm:composite}},
 \end{gathered}
\]
and
\[
 \begin{gathered}
 \text{trace model and subgroup rigidity}
 +\text{exact descent}
 +\text{bounded arithmetic reduction}\\
 +\text{the finite polynomial identity and multiplier argument}
 \quad\Longrightarrow\quad\text{Theorem~\ref{thm:thirty}}.
 \end{gathered}
\]
Both also use the prime-field and small-characteristic-range MDS theorems at the stated exclusion steps. Neither uses the classification assumption of Theorem~\ref{thm:standard}, the companion-orbit obstruction, or Fan's recent preprint. The composite-length theorem does not use Lemma~\ref{lem:finite-identity}. The finite-dimensional proof uses the independently established geometric ranges; it does not depend on the composite-length theorem.
\end{remark}

The scope of the composite-length conclusion can also be illustrated directly by its possible length factorizations.
\begin{example}
Theorem~\ref{thm:composite} applies to $n=\ell^a$ for every prime $\ell$ and every $a\geq2$, as well as to lengths with any number of prime factors. Its scope is irreducible cyclic codes in the check-polynomial sense of Definition~\ref{def:cyclic}; it does not assert the same theorem for arbitrary reducible cyclic codes.
\end{example}

\end{document}